\documentclass[12pt,english]{article}
\usepackage{mathptmx}
\usepackage{geometry}
\usepackage{color}
\usepackage{array}
\usepackage{bbding}
\usepackage{multirow}
\usepackage[fleqn]{amsmath}
\usepackage{amssymb}
\usepackage{amsthm}
\usepackage{graphicx}
\usepackage{natbib}
\usepackage{lscape}
\usepackage[hyphens]{url}
\usepackage[hidelinks]{hyperref}
\usepackage{comment}
\usepackage{bm}
\usepackage[title]{appendix}
\usepackage{longtable}
\usepackage{mathtools}
\usepackage{chngcntr}
\usepackage{authblk}
\usepackage{babel}
\usepackage{dsfont}
\usepackage{subcaption}
\usepackage{fancyhdr}
\usepackage{cleveref}
\usepackage{abstract}

\usepackage{fancyhdr}
\usepackage{cleveref}
\usepackage{makecell}
\usepackage{booktabs}
\usepackage{float}
\usepackage{xcolor}
\usepackage{algorithm}
\usepackage{algpseudocode}

\usepackage{threeparttable}
\usepackage{longtable}

\makeatletter
\allowdisplaybreaks
\date{}

\fancypagestyle{firstpage}{%
	\lhead{\footnotesize \textit{25th International Symposium on Transportation and Traffic Theory, ISTTT 25, 15-17 July 2024\\Ann Arbor, Michigan, United States}}
	\rhead{}
}

\makeatletter
\def\blfootnote{\xdef\@thefnmark{}\@footnotetext}
\makeatother

\usepackage[mathlines, pagewise]{lineno}
\usepackage{etoolbox} 

\newcommand*\linenomathpatchAMS[1]{%
	\expandafter\pretocmd\csname #1\endcsname {\linenomathAMS}{}{}%
	\expandafter\pretocmd\csname #1*\endcsname{\linenomathAMS}{}{}%
	\expandafter\apptocmd\csname end#1\endcsname {\endlinenomath}{}{}%
	\expandafter\apptocmd\csname end#1*\endcsname{\endlinenomath}{}{}%
}

\expandafter\ifx\linenomath\linenomathWithnumbers
\let\linenomathAMS\linenomathWithnumbers
\patchcmd\linenomathAMS{\advance\postdisplaypenalty\linenopenalty}{}{}{}
\else
\let\linenomathAMS\linenomathNonumbers
\fi

\linenomathpatchAMS{gather}
\linenomathpatchAMS{multline}
\linenomathpatchAMS{align}
\linenomathpatchAMS{alignat}
\linenomathpatchAMS{flalign}

\newtheorem{proposition}{Proposition}

\newtheorem{lemma}{Lemma}
\newtheorem{corollary}{Corollary}

\newtheorem{assumption}{Assumption}

\title{How would mobility-as-a-service (MaaS) platform survive as an intermediary? From the viewpoint of stability in many-to-many matching}
\author[a,b]{Rui Yao}
\author[a,$\ast$]{Kenan Zhang}
\affil[a]{EPFL, Switzerland}
\affil[b]{Technical University of Denmark, Denmark}

\DeclareMathOperator*{\R}{\mathcal{R}}

\DeclareMathOperator*{\A}{\mathcal{A}}
\DeclareMathOperator*{\N}{\mathcal{N}}
\DeclareMathOperator*{\M}{\mathcal{M}}
\DeclareMathOperator*{\W}{\mathcal{W}}
\DeclareMathOperator*{\setP}{\mathcal{P}}

\DeclareMathOperator*{\calL}{\mathcal{L}}

\begin{document}
\maketitle
\blfootnote{* Corresponding author. E-mail address: \href{mailto:kenzhang@ethz.ch}{kenan.zhang@epfl.ch}.}

\begin{abstract}
    Mobility-as-a-service (MaaS) provides seamless door-to-door trips by integrating different transport modes. Although many MaaS platforms have emerged in recent years, most of them remain at a limited integration level. 
    This study investigates the assignment and pricing problem for a MaaS platform as an intermediary in a multi-modal transportation network, which purchases capacity from service operators and sells multi-modal trips to travelers. 
    The analysis framework of many-to-many stable matching is adopted to decompose the joint design problem and to derive the stability condition such that both operators and travelers are willing to participate in the MaaS system.
    To maximize the flexibility in route choice and remove boundaries between modes, we design an origin-destination pricing scheme for MaaS trips. On the supply side, we propose a wholesale purchase price for service capacity. Accordingly, the assignment problem is reformulated and solved as a bi-level program, where MaaS travelers make multi-modal trips to minimize their travel costs meanwhile interacting with non-MaaS travelers in the multi-modal transport system. We prove that, under the proposed pricing scheme, there always exists a stable outcome to the overall many-to-many matching problem. Further, given an optimal assignment and under some mild conditions, a unique optimal pricing scheme is ensured. 
    Numerical experiments conducted on the extended Sioux Falls network also demonstrate that the proposed MaaS system could create a win-win-win situation---the MaaS platform is profitable and both traveler welfare and transit operator revenues increase from a baseline scenario without MaaS. 
\end{abstract}


\newpage
\section{Introduction}\label{sec:intro}
\noindent

Mobility-as-a-service (MaaS) is a new emerging concept for urban mobility. The key idea is to serve convenient and seamless door-to-door trips by integrating different transport modes. Accordingly, travelers buy mobility as a service instead of a bundle of travel modes~\citep{kamargianni2016critical,jittrapirom2017mobility}. 
While the integration of shared and on-demand mobility services with the mass transit network has been studied extensively in the literature~\citep[e.g.,][]{wang2016approximating,chen2017connecting,stiglic2018enhancing,pinto2020joint} and also been explored in real practice~\citep[e.g.,][]{shaheen2016mobility,wang2019new,mohiuddin2021planning}, the raise of MaaS brings up a series of open questions, e.g., who should play the role of MaaS platform? what is the relationship between the MaaS platform and mobility service operators? how should MaaS trips be priced and how should the revenue be distributed? could a MaaS system generate extra welfare for travelers and service operators meanwhile being profitable for the platform?

Motivated by the above questions, we study the assignment and pricing problem for a MaaS platform in a multi-model transportation network with both fixed-route transit and on-demand mobility services. Specifically, the MaaS platform is modeled as an intermediary that purchases capacity from the service operators and sells trips to travelers according to their origin-destination (OD) pairs, regardless of the modes or paths. Accordingly, MaaS travelers are provided with maximum flexibility in trip planning, which truly realizes the idea of MaaS (i.e., buying mobility as a single service rather than several means of transportation). In practice, the MaaS platform could recommend several trip plans based on different criteria (e.g., total travel time, number of transfers, etc.). Travelers could take any of them or follow other plans based on their own preferences.   
On the other hand, by wholesaling a fraction of capacity to the MaaS platform, the service operators also spread part of their risks. In other words, the service operators still operate as usual, but they now are guaranteed a certain level of revenue from the MaaS platform. 

This new design of the MaaS platform leads to several challenges. The first one is to assign passenger flows to the multi-modal traffic network, subject to the capacity constraints of mass transit as well as the congestion effect in on-demand mobility services. Different from classical multi-modal settings~\citep[e.g.,][]{di2019unified,luo2021multimodal}, our intermediary MaaS platform setting requires the assignment problem to jointly determine the capacity purchased from each service operator. 
Secondly, the proposed OD-based pricing brings additional challenges because it provides MaaS travelers with flexibility in their route choices. As will be shown later, it leads to complex equilibrium constraints in the assignment problem. 
The third key challenge regards pricing on both the demand and supply sides of the MaaS system. The main task here is to determine the capacity purchase price and the OD-based trip fare to maximize the MaaS platform's profit, meanwhile ensuring the participation of both service operators and travelers. 
To tackle these challenges, we adopt the analysis framework of many-to-many stable matching~\citep{gale1962college,sotomayor1992multiple}. 
Under stable matching outcomes, no other coalition of players can be created that achieves a higher total payoff. Hence, no one has an incentive to deviate from their current matching. To derive stable matching outcomes for the MaaS system, we first solve the multi-modal traffic assignment problem. Based on the optimal assignment, we then derive the stable conditions that retain the participation of travelers and service operators. These conditions, in turn, serve as constraints in the optimal pricing problem.  

Although many MaaS platforms have emerged in real-world (e.g., Transdev, Keolis, Moovel, Moovit), most of them remain at a limited integration level with a combined ticketing/payment system and/or trip planner~\citep{kamargianni2016critical}. While a few exceptions (e.g., Whim, UbiGo) sell mobility packages, this so-called MaaS bundle~\citep{reck2020maas} is still defined on specific services (e.g., one package contains certain numbers of taxi trips and bike-sharing trips). In contrast, we propose to completely remove the boundaries between transport modes through OD-based pricing. To the best of our knowledge, this MaaS pricing scheme has not been studied in the literature. Neither has it been implemented in real practice. 
This study also contributes to the explicit modeling of an intermediary MaaS platform. Different from previous studies~\citep[e.g.,][]{rasulkhani2019route,pantelidis2020many}, we introduce MaaS platform as another stakeholder in the market in addition to travelers and operations. It thus generates more insight into real practice. 
The proposed model further captures the interactions between MaaS and non-MaaS systems, which is missing from the literature~\cite[e.g.,][]{pantelidis2020many,deng2022incentive} but has a great impact on MaaS service planning. 

The rest of this paper is organized as follows. Section 2 briefly reviews related work. Section 3 presents the multi-modal network. The multi-modal traffic assignment and optimal pricing problem for MaaS many-to-many matching are presented in Section 4, along with discussions of solution existence and pricing uniqueness. Section 5 reports the results of the numerical experiment, with additional discussions in Section 6. Section 7 summarizes the main findings and comments on the directions for future research.

\section{Related work}\label{sec:review}
The research on the MaaS assignment and pricing problem is still limited in the literature. The closest ones to this study are \cite{pantelidis2020many} and \cite{liu2023demand}, both of which adopt the concept of many-to-many stable matching. While \cite{pantelidis2020many} only consider fixed-route transit, \cite{liu2023demand} incorporate the on-demand mobility services. In brief, the analysis of many-to-many stable matching consists of two steps~\citep{sotomayor1992multiple}. In the first step, a centralized assignment problem is solved that matches the supply of multiple ``sellers'' with the demand of multiple ``buyers". 
The second step is solving a constrained optimal pricing problem. Specifically, the constraints are given by the stability condition derived from the optimal assignment outcomes obtained in the first step, which guarantee that no pair of buyer and seller has the incentive to deviate from their current assignment and create a new matching. 
Accordingly, the joint assignment and pricing problem is disentangled into subproblems that are connected by the stability condition. 

This study also applies the analytical framework of many-to-many stable matching but differs in several ways. 
First and foremost, we explicitly model the MaaS platform as an intermediary that matches demand with supply and sets prices on both sides. In \cite{pantelidis2020many}, the supply-side decision is whether each service is activated in the MaaS network, and in \cite{liu2023demand}, multiple copies of the binary decision variables are created to represent different capacity levels. Differently, we consider exogenous mobility services and let the MaaS platform decide how much service capacity it purchases from each operator. The platform also needs to propose a capacity purchase price that sustains the participation of operators. Although largely complicating the stable conditions and the pricing problem, this modeling framework aligns much better with real practice. 
The second difference regards the demand-side pricing scheme. Here we consider OD-based pricing, which allows travelers to freely choose modes and paths to reach their destinations. As a result, the assignment problem is no longer fully centralized as in \cite{pantelidis2020many} and \cite{liu2023demand} but involves complex equilibrium constraints that describe the rational route choice of MaaS travelers. 
Thirdly, our model also captures the interactions between MaaS and non-MaaS travelers in the congestible transportation network, instead of assuming non-MaaS are teleported with fixed cost as in \cite{pantelidis2020many} and \cite{liu2023demand}. This prevents us from overestimating the benefit of MaaS and allows us to examine the impact of MaaS on the general public.

Using a stylized model with two service operators and four links, \cite{van2022business} analyze three types of MaaS platforms: 1) integrator, which sets the price for cross-mode trips and takes a share of revenue generated from MaaS trips; 2) broker\footnote{We do not use the term ``platform'' defined in \cite{van2022business} to avoid confusion with the MaaS platform considered in this study.}, which has no control on prices but only takes a shared of revenue of MaaS trips; and 3) intermediary, which first buys trips from the service operators and then sells them to MaaS travelers. It is found that travelers benefit the most from a MaaS integrator because in this case, the platform acts as an additional competitor in the mobility market. 
The MaaS platform modeled in this study is closest to the third case but is different in the way that the capacity purchase price is set by the MaaS platform rather than by the service operators. 
Another recent work also models the MaaS platform as an intermediary and adopts the auction theory to solve the matching problem between operators and travelers~\citep{ding2023mechanism}. In their setting, travelers bid for their trips without specifying the mode or path and the platform matches them to multi-model trips with the objective of maximizing social welfare. The service operators, on the other hand, fully comply with the trip assignment made by the platform. In other words, the supply-side problem is not explicitly considered in \cite{ding2023mechanism}.

This study is also closely related to the multi-modal traffic assignment models~\citep[e.g.,][]{di2019unified,chen2021ridesharing,yao2022general}. Specifically, \cite{deng2022incentive} model the stochastic user equilibrium in a hyperpath-based multi-modal network and solve the optimal pricing/incentive problem on top of that. It further discusses the profit-sharing scheme among service operators. Another array of research tackles the joint multi-modal network design and pricing problem~\citep[e.g.,][]{pinto2020joint,luo2021optimal,luo2021multimodal}. The main focus of this study, however, is to determine the optimal operational strategies for the MaaS platform, given the multi-modal network and total service capacity.

The MaaS platform considered in this study serves as an intermediary between service operators and travelers. Hence, the proposed model is also connected to the supply-chain models~\citep[e.g.,][]{nagurney2002supply,nagurney2006relationship}. Instead of characterizing the equilibria between competitive decision-makers in the supply-chain network, this study focuses on the optimal design of MaaS service under stable matching conditions between operators and travelers.

\section{Multi-modal network}\label{sec:setting}
Consider a multi-modal mobility network $G(\M, \N,\A)$ that is defined by the operator set $\M$, the node set $\N$ and the link set $\A$. 
$G$ consists of a series of subnetworks $G_m$, each of which denotes the service network of an operator $m\in\M$, the sets of origin and destination nodes, denoted by $\N_O$ and $\N_D$, respectively, a set of dummy links $\A_\text{dm}$ that connect origin/destination nodes to the service networks, and a set of transfer links $\A_\text{tr}$ that connect between different service networks. 
An example of the multi-modal mobility network is illustrated in Figure~\ref{fig:MaaS_net}.

\begin{figure}[htb]
    \centering
    \includegraphics[width=0.8\textwidth]{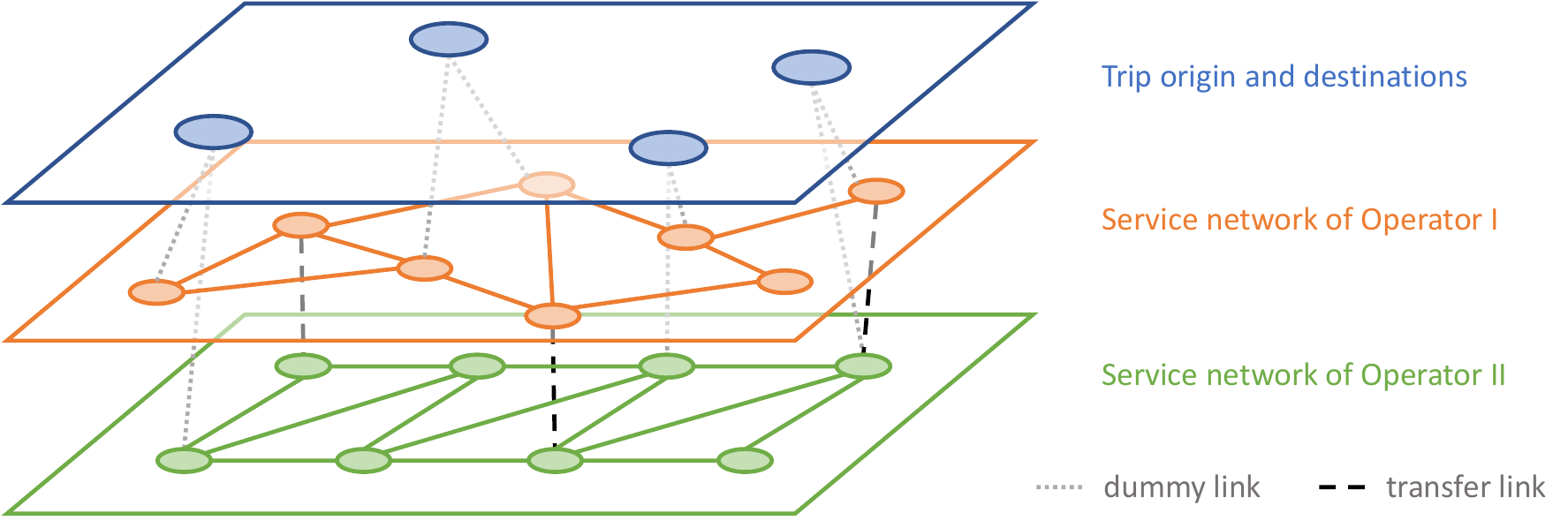}
    \caption{A multi-modal mobility network.}
    \label{fig:MaaS_net}
\end{figure}

For each service network, the links are further divided into regular links, which represent regular trip segments; and access links, which capture the queuing/searching process in a specific type of mobility service. In this study, we consider two categories of mobility services, namely, mass transit (MT), which includes fixed-route bus and railway, and mobility-on-demand service (MoD), such as ride-hailing and bike-sharing. Accordingly, there are six subsets of service links: MT regular links $\A_\text{MT-r}$, MT access links $\A_\text{MT-a}$, MT egress links $\A_\text{MT-e}$, MoD regular links $\A_\text{MoD-r}$, MoD access links $\A_\text{MoD-a}$, and MoD egress links $\A_\text{MoD-e}$. 
Besides service networks, there is also a subnetwork corresponding to private driving, an optional mode for non-MaaS travelers but not for MaaS travelers. The set of links is denoted by $\A_\text{Drv}$.

In what follows, we specify the travel time for each type of link in the multi-modal network.

\subsection{MT links}
The link travel time in mass transit is often assumed to be constant~\citep{de1993transit,LAM2002919,szeto2014transit}, and since we consider exogenous mobility services in this paper, it is safe to assume the MT has fixed headway, which yields a fixed travel time on the access links. Additionally, the travel time on egress links is also assumed to be constant. Accordingly, 
\begin{align}\label{eq:ta_MT}
    t_a = T_a,\quad  a\in{\A}_\text{MT-r}\cup {\A}_\text{MT-a}\cup{\A}_\text{MT-e}.
\end{align}
Note that, we assume explicit capacity constraints on MT regular links (detailed in the next section). Similar to \cite{LAM2002919}, the congestion effects in MT will be captured implicitly by the dual variable associated with the capacity constraints.

\subsection{MoD links}
Consider ride-hailing and bike-sharing as the two typical MoD services. The former is operated on the road network whereas the latter can be considered independent of the road traffic. 
Hence, we define two sets of MoD regular links, denoted by $\A_\text{MoD-r1}$ and $\A_\text{MoD-r2}$, and accordingly divide the MoD operators into two groups, denoted by $\M_\text{MoD-1}$ and $\M_\text{MoD-2}$. The travel time on link $a\in \A_\text{MoD-r1}$ is jointly determined by its own flow, the flow of other ride-hailing services, and the background traffic.
Following the convention, we model it with the Bureau of Public Roads (BPR) function:
\begin{align}
    t_a = T_a \left(1+0.15\left(\frac{\sum_{m\in\M_\text{MoD}}(x_{a,m}+\tilde{x}_{a,m})+\bar{x}_a}{K_a}\right)^4\right),\quad a\in{\A}_\text{MoD-r1},\label{eq:ta_MoD_regular_type1}
\end{align}
where $x_{a,m}, \tilde{x}_{a,m}$ are the link flows of MaaS travelers and non-MaaS travelers, respectively, $,\bar{x}_a=\tilde{x}_{a'},a'\in\A_\text{Drv}$ is the link flow of background traffic corresponding to the same road link, $T_a$ is the free-flow travel time, and $K_a$ is the fixed road capacity. 

On the other hand, the travel time on link $a\in \A_\text{MoD-r2}$ can be simply set to a constant and equal among the same type of MoD services. 
\begin{align}
    t_a = T_a,\quad a\in{\A}_\text{MoD-r2},\label{eq:ta_MoD_regular_type2}
\end{align}

The travel time of an access link refers to the passenger waiting in the MoD service. 
For ride-hailing services, we introduce the following approximation to simplify the analysis:
\begin{align}
    t_{a,m} &= \kappa_m \frac{\sum_{a\in\A_\text{MoD-a}} (x_{a,m} + \tilde{x}_{a,m})}{V_m},\quad a\in{\A}_\text{MoD-a},\;m\in{\M}_\text{MoD-1},\label{eq:ta_MoD_access}\\
    V_m &= K_m - \sum_{a\in\A_\text{MoD-r1}}t_a(x_{a,m} + \tilde{x}_{a,m})
\end{align}
where the parameter $\kappa_m$ summarizes the impact of factors other than the demand-supply relationship (e.g., matching efficiency), and $K_m$ is the fleet size of MoD operator $m$. 

Eq.~\eqref{eq:ta_MoD_access} is derived from the Cobb-Douglas matching function \citep{cobb1928theory} with the scaling factor equal to 1/2 for both demand and supply. Intuitively, it states that the waiting time is uniform over the network and is proportional to the demand-supply ratio. Specifically, the numerator $\sum_{a\in\A^m_\text{MoD-a}} x_{a,m} + \tilde{x}_{a,m}$ gives the total number of travelers that join the MoD service, while the denominator $V_m = K_m - \sum_{a\in\A_\text{MoD-r1}}t_a(x_{a,m} + \tilde{x}_{a,m})$ refers to the idle fleet that are available for matching. 

Since the waiting time in a dockless bike-sharing service shares the same property as ride-hailing~\citep{zheng2021many}, we may use Eq.~\eqref{eq:ta_MoD_access} to estimate its access time as well, while replacing the regular link set with $\A_\text{MoD-r2}$. 

Similar to MT, we assume the MoD egress links have constant travel time, i.e., 
\begin{align}
    t_a = T_a,\quad a\in{\A}_\text{MoD-e}.
\end{align}

\subsection{Dummy and transfer links}
We assume the dummy links induce no travel time, while each mode transfer takes a fixed travel time. Hence, 
\begin{align}
    t_a &= 0,\quad a\in{\A}_\text{dm},\\
    t_a &= T_0,\quad a\in{\A}_\text{tr},
\end{align}
where $T_0$ is the fixed mode transfer time.

\section{MaaS as many-to-many matching}\label{sec:model}
\subsection{System overview}\label{sec:overview}
Consider a MaaS platform that purchases service capacity from the operators and offers multi-modal trips to travelers. To facilitate the analysis, we introduce the following assumptions.
\begin{assumption}\label{assumption}
    Standing assumptions of the MaaS system:
    \begin{enumerate}
        \item The operator's total service capacity and non-MaaS trip fares are fixed and exogenous;
        \item Total travel demand is fixed and exogenous;
        \item MaaS travelers can freely choose their routes to minimize total travel disutility, or equivalently, total travel time under OD-based pricing;
        \item Non-MaaS travelers may take a self-designed multi-modal trip but are subject to a non-negative planning cost for each mode transfer. 
        Besides, they opt for driving when the service capacity in the multi-modal mobility network is not sufficient to support their trips. 
    \end{enumerate}
\end{assumption}

As described in Section~\ref{sec:intro}, the assignment and pricing problem for the MaaS platform can be formulated as a variant of many-to-many stable matching problem~\citep{sotomayor1992multiple}. Specifically, we consider each OD pair as a ``buyer'' and each service operator as a ``seller''. On the one hand, the MaaS platform matches each OD pair to multiple service operators through multi-modal trips. On the other hand, it purchases the service capacity from each service operator to serve travelers between multiple OD pairs. Accordingly, the MaaS platform matches travel demand and service capacity on the multi-modal paths. As per the last point in Assumption~\ref{assumption}, non-MaaS travelers may also take multi-modal trips. Hence, the remaining travel demand and service capacity are matched on the multi-modal paths as well, though neither the assignment nor the fare is decided by the MaaS platform. 

Figure~\ref{fig:matching} illustrates the many-to-many matching for a simple network with one OD pair and three operators (i.e., transit, ride-hailing, and bike-sharing). Travelers can choose among three paths, each of which consists of links from one or more operators. For simplicity, here we do not plot the transfer or access links. Neither do we distinguish MaaS and non-MaaS demands. Yet, they will be considered in the assignment problem formulated in the next subsection. 

\begin{figure}[htb]
    \centering
    \includegraphics[width=0.8\textwidth]{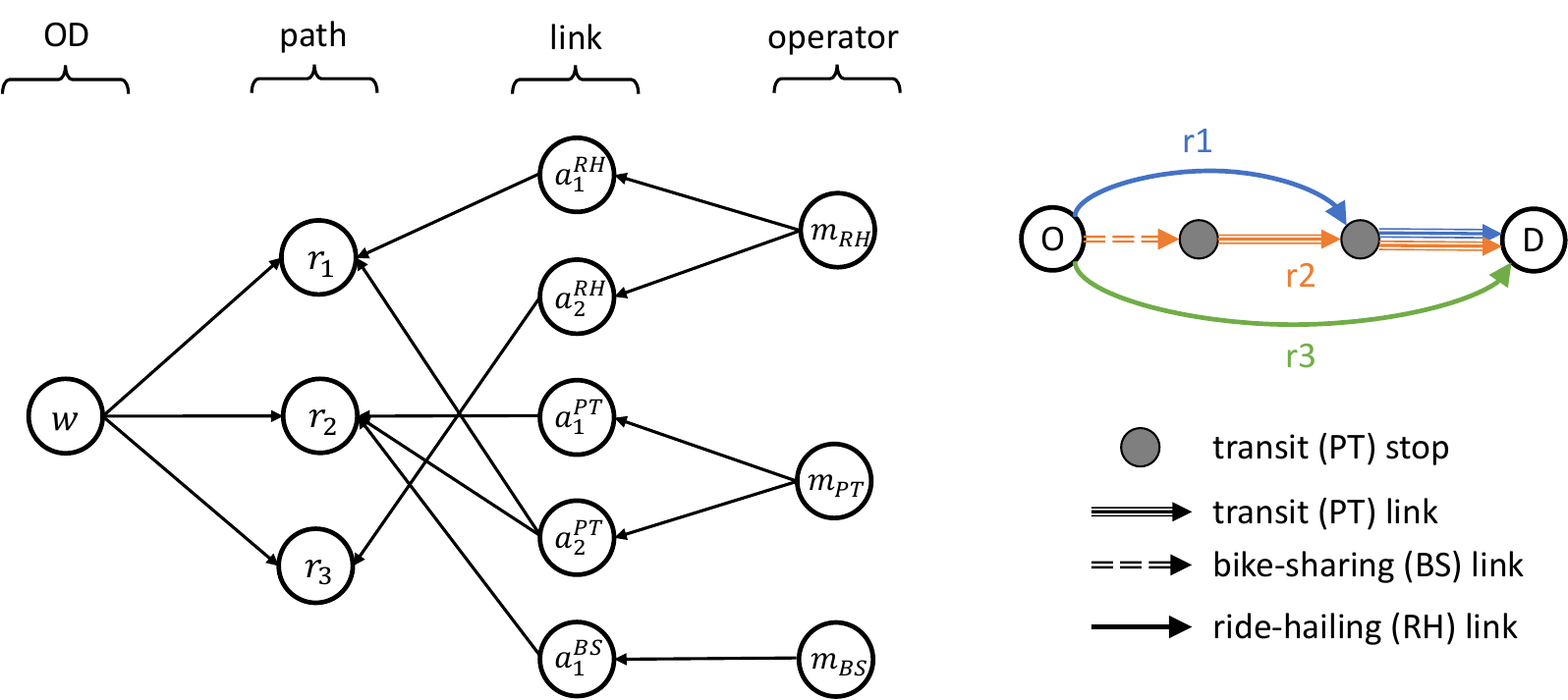}
    \caption{Matching on the multi-modal network.}
    \label{fig:matching}
\end{figure}

\subsection{Assignment problem}\label{sec:assignment}
Let $\W$ denote the set of OD pairs and $Q_w$ be the total travel demand between OD pair $w\in\W$. 
For MT operators, we use $K_a$ to denote the total capacity of each regular link $a\in\A_\text{MT-r}$. 
As for MoD operators, the service capacity on each regular link is the same as link flow, while the total vehicle time, either occupied or vacant, should equal a fixed fleet size $K_m$. 
In this study, we do not explicitly model the vehicle rebalancing. Instead, we introduce a constraint $V_m>\underline{V}_m$ with a minimum vacant vehicle time $\underline{V}_m>0$ to guarantee a certain level-of-service.

In words, the assignment problem is to determine the MaaS demand $q_w$, link capacity purchased $k_a$ and resulting link flows $x_a$ such that the total cost in the multi-modal mobility system is minimized. The total cost consists of two parts: the total travel time of travelers and the total operation cost of mobility services. All other costs, e.g., trip fares and capacity purchase price, are intermediate transfers between stakeholders within the system and thus not counted. Since the mobility services are assumed to be exogenous in this study, their operation cost is constant and thus can be dropped from the objective function. 
To simplify the formulation, we introduce node potential $\pi_{ij}$ to denote the minimum travel cost from node $i$ to node $j$ and define $x^w_a$ and $\tilde{x}^w_a$ as the MaaS and non-MaaS link flows associated with OD pair $w\in\W$, respectively.

In what follows, the accent ``$\sim$'' is added to all non-MaaS variables. Specifically, $\tilde{p}_a,\;a\in\A_\text{MT-r}\cup\A_\text{MoD-r}$ denotes the exogenous link-based fare set by the service operators. To capture the extra burden of multi-modal trip planning for non-MaaS travelers, we further impose a fixed cost $P_0$ to each transfer link, i.e., $\tilde{p}_a=P_0,\;\forall a\in\A_\text{tr}$.

The assignment problem is formulated as follows:
\begin{subequations}\label{eq:assignment}
\begin{align}
    \min_{q,k,x,\pi}\quad & \sum_{a\in\A} t_a(x_a+\tilde{x}_a) \label{eq:obj_assign}\\
    s.t.\quad & q_w + \tilde{q}_w = Q_w, &&\; \forall w\in\W, \label{eq:demand}\\
    & V_m + \sum_{a\in\A_\text{{MoD-r}}} t_a(x_{a,m} + \tilde{x}_{a,m}) = K_m, &&\; \forall m\in{\M}_\text{MoD},\label{eq:MoD_cap}\\
    & V_m \geq \underline{V}_m, &&\; \forall m\in{\M}_\text{MoD},\label{eq:MoD_minvac}\\
    & 0\leq \left(\pi_{js} + t_a - \pi_{is}\right) \perp x^w_a \geq 0, &&\; \forall a:(i,j)\in\A,\; w:(r,s)\in\W, \label{eq:UE_MaaS}\\
    & 0\leq \left(\tilde{\pi}_{js} + t_a + \tilde{p}_a - \tilde{\pi}_{is}\right) \perp \tilde{x}^w_a \geq 0, &&\; \forall a:(i,j)\in\A,\; w:(r,s)\in\W, \label{eq:UE_nonMaaS}\\
    & \sum_{w\in\W} x^w_a = x_a,\; \sum_{w\in\W} \tilde{x}^w_a = \tilde{x}_a, &&\;\forall a\in\A, \\
    & x \in \Omega(q),\; \tilde{x}\in\Omega(\tilde{q}),\label{eq:feasible}\\
    & x_a \leq k_a, \; \tilde{x}_a \leq K_a - k_a &&\;  \forall a\in{\A}_\text{MT-r},\label{eq:link_cap}\\
    & q, \tilde{q}, k, \tilde{k}, \pi, \tilde{\pi}\geq 0.\label{eq:nonneg}
\end{align}
\end{subequations}

Constraint~\eqref{eq:demand} states that the demand between each OD is split into MaaS and non-MaaS travelers.
Constraint~\eqref{eq:MoD_cap} is the conservation of MoD fleet and Constraint~\eqref{eq:MoD_minvac} guarantees their service qualities. Note that Constraint~\eqref{eq:MoD_cap} also implies the MoD service capacity purchased by the MaaS platform matches exactly the demand it serves (i.e., $k_{a,m} = x_{a,m}, \forall a\in\A_\text{{MoD-r}},m\in\M_\text{MoD}$). 

Constraints~\eqref{eq:UE_MaaS} and \eqref{eq:UE_nonMaaS} describe the UE conditions for MaaS and non-MaaS travelers. 
The former is due to the OD-based pricing scheme (i.e., MaaS travelers can freely choose routes between their OD pair), while the latter represents the routing behaviors outside the MaaS system. These two complementary constraints do not exist in the classic stable matching problem~\citep{gale1962college,sotomayor1992multiple}. 
In \cite{pantelidis2020many} and \cite{liu2023demand}, the non-MaaS travelers are simply teleported to their destinations with a fixed cost. Therefore, our formulation is more general and correctly captures the interactions between MaaS and non-MaaS travelers in the multi-modal mobility network.

Finally, $\Omega(q)$ and $\tilde{\Omega}(\tilde{q})$ respectively define the feasible sets of MaaS and non-MaaS OD-indexed link flows that satisfy the flow conservation constraints. For instance, any feasible link flow $x\in\Omega(q)$ satisfies
\begin{align}\label{eq:feasible_set}
    &\sum_{k:(k,i)\in\A} x^w_{ki} - \sum_{j:(i,j)\in\A} x^w_{ij} = \begin{cases}
      -q_w, & i=R(w) \\
      q_w, & i=S(w) \\ 
      0, & \text{otherwise}
    \end{cases}, \;\forall i\in\N,w\in\W,\\
    &x^w_{a}\geq 0, \; \forall w \in \W, a \in \A\setminus{\A}_\text{Drv},
\end{align}
where $R(w)$ and $S(w)$ denote the origin and destination of OD pair $w$, respectively. $\tilde{\Omega}$ is only different in the set of links, which also includes private driving links. 
The MT link capacity constraints are dictated in Eq.~\eqref{eq:link_cap}, which, together with the non-negative constraints~\eqref{eq:nonneg}, imply the conservation of MT total service capacity.

\subsubsection{Reformulation as a bi-level program}\label{sec:bilevel}
The assignment problem \eqref{eq:assignment} is not easy to solve due to the equilibrium constraints \eqref{eq:UE_MaaS} and \eqref{eq:UE_nonMaaS}, especially in real-size networks. Hence, we proposed to reformulate the problem into a bi-level program, where the MaaS demand and capacity planning is solved at the upper level and the lower-level problem turns out to be a traffic assignment problem with capacity constraints. 
In this section, we present the bi-level formulation, along with some analytical results to simplify the original problem, while the detailed solution algorithm is presented Section~\ref{sec:algo}.

We first observe that a profit-maximizing MaaS platform would purchase the exact service capacity as the realized traffic flow. This result is presented in the following lemma. 

\begin{lemma}[Market clearance]\label{lemma:market-clearance}
Suppose the MaaS platform aims to maximize its profit. Then, at the optimal assignment, the purchased service capacity equals the realized MaaS travel flows at optimum, i.e., $k^*_a=x^*_a, \forall a\in\A_\text{MT-r}$.
\end{lemma}
\begin{proof}
    See Appendix \ref{proof:market_clearance} \citep[modified from][]{nagurney2002supply}.
\end{proof}

Accordingly, we can safely drop the decision variable $k_a$ and reformulate the problem as a bi-level program, where the upper-level problem solves the optimal MaaS demand allocation and the lower-level problem solves a multi-modal traffic equilibrium with capacity constraints. The lower-level problem can be further turned into a variational inequality (VI) problem. 
The simplified bi-level formulation is formally stated in Proposition~\ref{prop:bilevel}.

\begin{proposition}[Bi-level formulation of optimal MaaS assignment]\label{prop:bilevel}
    For a profit-driven MaaS platform, the assignment problem ~\eqref{eq:assignment} is equivalent to the following bi-level program:\\
    \noindent \textbf{[Upper-level problem]} 
    \begin{subequations}\label{eq:upper_simp}
    \begin{align}
        \min_{q}\quad & \sum_{a\in\A} t^*_a(x^*_a + \tilde{x}^*_a) \label{eq:upper_obj_simp}\\
        s.t.\quad & 0 \leq q_w \leq Q_w,\; \forall w\in\W,\label{eq:upper_simp_demand_constraint}\\
        & x^*, \tilde{x}^* \in S(q), 
    \end{align}
    \end{subequations}
    where $S(q)$ is the solution set of the lower-level problem.

    \noindent\textbf{[Lower-level problem]} 
    Let $z=[x, \tilde{x}]$ and $F(z) = [t(z), t(z)+\tilde{p}]$. 
    Solving $z^*$ such that
    \begin{subequations}\label{eq:lower_simp_whole}
        \begin{align}\label{eq:lower_simp}
        \langle F(z^*), z-z^*\rangle \geq 0,\; \forall z\in\Gamma(q),
    \end{align}
    where the feasible set $\Gamma(q)$ is defined on the following constraints:
    \begin{align}
        &x\in\Omega(q), \; \tilde{x}\in\tilde{\Omega}(Q-q), \label{eq:lower_simp_flow_feasibility} \\
        &\sum_{a\in\A_\text{MoD-r}} t_a(x_{a,m} + \tilde{x}_{a,m})  \leq K_m - \underline{V}_m,\; \forall m\in{\M}_\text{MoD}, \label{eq:lower_simp_MoD_constraint}\\
        &x_a + \tilde{x}_a \leq K_a,\; \forall a\in{\A}_\text{MT-r}.\label{eq:lower_simp_MT_constraint}
    \end{align} 
    \end{subequations}    
\end{proposition}
\begin{proof}
    Due to Lemma~\ref{lemma:market-clearance}, the MT capacity constraints~\eqref{eq:link_cap} can be equivalently substituted by constraint~\eqref{eq:lower_simp_MT_constraint}. 
    The non-MaaS demand flow $\tilde{q}$ can be dropped by rewriting Eqs.~\eqref{eq:demand} and \eqref{eq:feasible} into Constraints~\eqref{eq:upper_simp_demand_constraint}, \eqref{eq:lower_simp_flow_feasibility}. 
    Then, the MoD capacity constraints~\eqref{eq:MoD_cap} and \eqref{eq:MoD_minvac} can be combined into Constraint~\eqref{eq:lower_simp_MoD_constraint}.
    Lastly, the complementarity constraints~\eqref{eq:UE_MaaS}-\eqref{eq:UE_nonMaaS} are substituted by $S(q)$, which corresponds to the set of solutions to the equivalent VI problem~\eqref{eq:lower_simp}. This completes the proof.
\end{proof}

To further simplify the lower-level problem~\eqref{eq:lower_simp}, we define $x^r_a=\sum_{w:(r,\cdot)}x^w_a$ as the MaaS link flow originated from node $r\in\N_O$ and $\tilde{x}^r_a$ in the same way. Then, let $A=(A^r)_{r\in\N_O}$ be $\N_O$ stacking node-link incidence matrices, where $A^r_{i,a} = \begin{cases}
    -1, & a = (i,\cdot)\\
    1, & a = (\cdot, i) \\
    0, & \text{otherwise}
\end{cases}$, and define vector $b$ with each element $b^r_i = \begin{cases}
    \sum_{w:(r,\cdot)} -q^w, & i = r\\
    q^w, & i = s, w:(r,s)\\
    0 & \text{otherwise}
\end{cases}$ and $\tilde{b}$ is define in the same way with $\tilde{q}$. 
Accordingly, $\Omega(q)$ can be rewritten as  $\Omega(q) = \{x| Ax = b, x\geq 0\}$ and similarly for $\tilde{\Omega}(Q-q)$. 

\begin{proposition}[Existence of optimal MaaS assignment]\label{prop:existence_assignment}
    Suppose the MoD capacities $K_m, \forall m\in\M_\text{MoD}$ are sufficient and the MT capacities $K_a, \forall a\in\A_\text{MT-r}$ are reasonably large such that $\Gamma(Q)$ is nonempty, i.e., there exists a feasible link flow pattern $x$ when MaaS captures all demand in the system.
    Then, there always exists a solution $(q^*, x^*)$ to the assignment problem~\eqref{eq:upper_simp}-\eqref{eq:lower_simp_whole}.
\end{proposition}
\begin{proof}
    When the MoD capacities are sufficient, Constraint~\eqref{eq:lower_simp_MoD_constraint} always holds and thus can be dropped safely. 
    For any MaaS demand $q$, it is easy to show $\Omega(q)$ and $\tilde{\Omega}(Q-q)$ are nonempty, convex, and compact. To further show $\Gamma(q)$ is indeed nonempty, convex, and compact, we recall that $\Gamma(Q)$ is non-empty due to the MT capacity assumption. Therefore, starting from a feasible solution in $x_0 \in\Gamma(Q)$, for any MaaS demand $q$, we can construct a feasible link flow pattern $(x,\tilde{x})$ such that $x\in\Omega(q),\tilde{x}\in\Omega(\tilde{q})$ and $x_a+\tilde{x}_a = x_{0,a},\forall a\in\A_\text{MT-r}$. Since $\tilde{\Omega}$ is defined on a larger set of links that include $\A_\text{Drv}$, we have $\Omega(\tilde{q})\subseteq \tilde{\Omega}(\tilde{q})$ for any non-MaaS demand $\tilde{q}$. It thus ensures that $\tilde{x}\in\tilde{\Omega}(Q-q)$. Accordingly, we conclude $\Gamma(q)$ is nonempty, convex, and compact for any $q$. 
    
    Since the half-space defined by Constraint~\eqref{eq:lower_simp_MT_constraint} is defined on the total link flow, it is affected by total demand $Q$ but not $q$. On the other hand,  $\Omega(q)$ and $\tilde{\Omega}(Q-q)$ change continuously with $q$. Therefore, $\Gamma(q)$ is a continuous point-to-set map of $q$. 

    We then prove the existence of optimal assignment by evoking Corollary 3 in \cite{harker1988existence}. Besides the property of $\Gamma(q)$ stated above, it requires three more conditions: 1) the upper-level objective function is continuous in both $q$ and $x^*$; 2) the lower-level cost function $F(z)$ is continuous; and 3) the feasible set of $q$ is nonempty and compact. All of them are trivially satisfied and thus the proof is complete.
\end{proof}
We note that the assumptions in Proposition~\ref{prop:existence_assignment} are reasonable. First, one primary goal of introducing MaaS is to increase public transport ridership. Hence, both MoD and MT should have sufficient capacity to serve demand switching from private driving. Besides, as the total demand is fixed, the link flow is bounded, i.e., $x_a\leq \sum_w Q_w,\forall a\in\A_\text{MoD-r}$. Therefore, one sufficiently large fleet size can be set as $K_m = \sum_{a\in\A_\text{MoD-rj}} t_a(\sum_w Q_w)\sum_w Q_w  + \underline{V}_m$.

\subsubsection{Solution algorithm}\label{sec:algo}

In this study, we adopt a gradient-based algorithm to solve the assignment problem \eqref{eq:upper_simp}-\eqref{eq:lower_simp_whole}. Let $\calL(q)$ denote the objective function \eqref{eq:upper_obj_simp}, then the gradient is given by 
\begin{align}\label{eq:assignment_obj_grad}
     \nabla \calL(q) &= \sum_{a\in\A} \left[ (x_a^*+\tilde{x}^*_a)\frac{\partial t^*_a}{\partial x_a} + t^*_a\right]\nabla_q x^*_a + \sum_{a\in\A} \left[ (x_a^*+\tilde{x}^*_a)\frac{\partial t^*_a}{\partial \tilde{x}_a} + t^*_a\right]\nabla_q \tilde{x}^*_a \nonumber\\
     &= \sum_{a\in\A} \left[ (x_a^*+\tilde{x}^*_a)(t^*_a)' + t^*_a\right](\nabla_q x^*_a + \nabla_q \tilde{x}^*_a)
\end{align}
where $\nabla_q x^*_a$ and $\nabla_q \tilde{x}^*_a$ are the gradients of $x^*_a$ and $\tilde{x}^*_a$ with respect to $q$, respectively, also known as the equilibrium sensitivities \citep{tobin1988sensitivity,patriksson2004sensitivity}. The second equality is due to the fact that all link travel times can be seen as a function of total link flow $x_a+\tilde{x}_a$ and thus $\partial t_a/\partial x_a = \partial t_a/\partial \tilde{x}_a = t'_a$.

To evaluate the gradient Eq.~\eqref{eq:assignment_obj_grad}, we first consider a relaxed lower-level problem 
\begin{align}\label{eq:VI_relax}
    \langle \hat{F}(z^*), z-z^*\rangle \geq 0,\; \forall z\in\Phi(q). 
\end{align}
Here, the feasible set $\Phi(q) = \{z| \Lambda z = d, z\geq 0\}$, where $\Lambda = [A, A]$ and $d=[b, \tilde{b}]$, only contains affine function and non-negative constraint. The MT capacity constraint \eqref{eq:lower_simp_MT_constraint} is relaxed and added to the cost function
\begin{align}\label{eq:augment_link_cost}
    \hat{t}_a = t_a + [\mu_a + \rho(x_a + \tilde{x}_a - K_a)]_+,\;\forall a \in{\A}_\text{MT-r},
\end{align}
where $\mu_a$ is the Lagrangian multiplier and $\rho$ is a positive penalty factor. They can be seen as constants for a single lower-level problem, though they are updated iteratively in the main loop. The augmented link cost \eqref{eq:augment_link_cost} is motivated by the augmented Lagrangian method also used \citep{nie2004models,kanzow2016multiplier} to tackle the side constraints in traffic assignment problems.
Accordingly, we have $\hat{F}(z) = [\hat{t}(z), \hat{t}(z) + \tilde{p}]$. The MoD capacity constraint \eqref{eq:lower_simp_MoD_constraint}, on the other hand, is naturally satisfied because a violation of such constraint would lead to an infinite access time. In other words, the demand flow for MoD service would not keep growing because the decreasing vacant vehicle time $V_m$ deteriorates the service quality and makes the service less attractive to travelers. 

It is a well known result \citep[e.g.,][]{facchinei2003finite} that the VI problem \eqref{eq:VI_relax} is equivalent to a fixed point
\begin{align}\label{eq:fix_point}
    z^* = \Pi_{\Phi(q)}(z - \gamma \hat{F}(z^*)), 
\end{align}
for any step size $\gamma>0$, where $\Pi$ is the projection operator. Hence, we may solve the lower-level traffic equilibrium through fixed-point iterations and differentiate the result through backward propagation~\citep{li2022differentiable,liu2023end}. Yet, evaluating and differentiating the projection operation $\Pi_{\Phi(q)}$ is still challenging. To tackle this issue, we apply the operator splitting method following \cite{davis2017three}. Let $\Phi_1 = \{z|z\geq 0\}$ and $\Phi_2 = \{z|\Lambda z = d\}$, then the projections onto $\Phi_1$ and $\Phi_2$ have the following closed-form solutions:
\begin{align}
    \Pi_{\Phi_1}(z) &= [z]_+,\\
    \Pi_{\Phi_2}(z) &= z - U\Sigma^{-1}V^\top (\Lambda z - d),
\end{align}
where $U\Sigma V$ is the compact singular value decomposition of the adjacency matrix $\Lambda$. 

Accordingly, we solve Eq.\eqref{eq:VI_relax} through the following fixed-point iterations with auxiliary variables $u$ and $v$:
\begin{subequations}\label{eq:fix_point_iter}
    \begin{align}
        z^{(n+1)} &= \Pi_{\Phi_1}(u^{(n)}),\\
        v^{(n+1)} &= \Pi_{\Phi_2}\left(2z^{(n+1)} - u^{(n)} - \gamma \hat{F}(z^{(n+1)})\right),\\
        u^{(n+1)} &= u^{(n)} + \left(v^{(n+1)} - z^{(n+1)}\right)
    \end{align}
\end{subequations}

\begin{algorithm}[H]
\caption{Solution algorithm for MaaS assignment problem}\label{alg:assignment}
\begin{algorithmic}
    \State \textbf{Parameter:} step size $\alpha,\gamma$, momentum parameter $\beta$, initial penalty factor $\rho$, penalty scaler $\phi$, penalty threshold parameter $\sigma$, gap thresholds $\varepsilon_q,\varepsilon_z$
    \State \textbf{Input:} initial MaaS demand $q_0$, initial link flows $x_0, \tilde{x}_0$, maximum iteration for lower-level problem $N$
    \State \textbf{Initialization:} $q^{(0)} = q_0$, $z^{*(0)}=(x_0, \tilde{x}_0)$, $\omega^{(0)}=0$, $\mu^{(0)}=0$, $\rho^{(0)} = \rho$
    \For{$k=0,1,2,\dots$}
        \State Run $N$ fixed-point iterations \eqref{eq:fix_point_iter} with $u^{(0)} = z^{(k)}$.
        \State Set $z^{(k+1)}= \Pi_{\Phi_1}(u^{(N)}) = [x^{(k+1)}, \tilde{x}^{(k+1)}]$.
        \State Compute Jacobian matrix $J_q(z^{(k+1)})$. Accordingly, $\nabla_q x^{(k+1)}_a$ is the sum over rows in $J_q(z^{(k+1)})$ that corresponds to MaaS link flow $x^r_a,\forall r\in\N_O$, and $\nabla_q x^{(k+1)}_a$ is obtained in the same way. 
        \State Compute gradient $\nabla \calL(q^{(k)})$ by Eq.~\eqref{eq:assignment_obj_grad}.
        \State Update MaaS demand 
        \begin{align}
            \omega^{(k+1)} &= \beta \omega^{(k)} - \alpha \nabla \calL(q^{(k)})\\
            q^{(k+1)} &= \min \left\{Q, \left[q^{(k)} - \alpha \nabla \calL(q^{(k)}) + \beta \omega^{(k+1)} \right]_+ \right\}
        \end{align}
        \State Update Lagrangian multipliers: $\forall a\in\A_\text{MT-r}$,
        \begin{subequations}\label{eq:update_mu}
            \begin{align}
                \mu_a^{(k+1)} &= \left[\mu_a^{(k)} + \rho^{(k)}\left(x_a^{(k+1)} + \tilde{x}_a^{(k+1)} - K_a \right) \right]_+\\
                \rho^{(k+1)} &= \begin{cases}
                    \phi +  \rho^{(k)}, \; |x_a^{(k+1)} + \tilde{x}_a^{(k+1)} - K_a| \geq \sigma |x_a^{(k)} + \tilde{x}_a^{(k)} - K_a|\\
                    \rho^{(k)},\;\text{otherwise}
                \end{cases}
            \end{align}
        \end{subequations}
        
        \State Compute gaps
        \begin{align}
            g_q^{(k+1)} = \frac{||q^{(k+1)} - q^{(k)}||}{||q^{(k)}||}, \quad 
            g_z^{(k+1)} = \frac{||z^{(k+1)} - z^{(k)}||}{||z^{(k)}||}
        \end{align}
        \If{$g_q^{(k+1)}<\varepsilon_q$, $g_z^{(k+1)}<\varepsilon_z$, and capacity constraints are satisfied}
            \State \textbf{break}
        \EndIf
        \State Update $N$
    \EndFor
\end{algorithmic}
\end{algorithm}

The convergence of the fixed-point iterations is proved in the following proposition:

\begin{proposition}[Convergence of fixed-point iterations with operator splitting]\label{prop:lower_level_convergence} 
    Then, the fixed-point iterations \eqref{eq:fix_point_iter} converges linearly to the solution to the relaxed VI problem \eqref{eq:VI_relax} with a step size $\gamma\in(0, 2/(L+\rho))$, where $L$ is the Lipschitz constant of link travel time function $t(\cdot)$. 
\end{proposition}
\begin{proof}
    We first note that $t(\cdot)$ must be Lipschitz continuous because link flows are bounded by the fixed travel demand. 
    Correspondingly, the Lipschitz constant of $\hat{F}$ is $L+\rho$, with $\rho$ being the augmented Lagrangian penalty.
    Furthermore, $\hat{F}(z)$ is maximal monotone as it is monotone and continuous with $z$~\citep{auslender2003maximal}. 
    The remaining proof can be found in \cite{davis2017three}.
\end{proof}

Algorithm~\ref{alg:assignment} summarizes the solution algorithm for the MaaS assignment problem. In brief, in each iteration, we first solve the approximate traffic equilibrium given the current MaaS demand and compute its Jacobian matrix. Then, we update the MaaS demand via Nesterov accelerated gradient descent \citep{sutskever2013importance} and the Lagrangian multipliers \citep{parikh2014proximal}. As it approaches the optimal solution, we gradually reduce the number of fixed-point iterations by 1, which is numerically shown to stabilize the outcomes. 
The iterations terminate when both demand gap $g_q$ and flow gap $g_z$ are sufficiently small.

\subsection{Pricing problem}\label{sec:pricing}
The pricing problem presented in this section aims to determine the capacity purchase price, denoted by $p^s_a,\;\forall a\in\A_\text{MT}\cup\A_\text{MoD}$, and the trip fare between each OD pair, denoted by $p^d_w,\;w\in\W$, such that the matching, or equivalently the flow assignment, solved in Section~\ref{sec:assignment} satisfies the stable condition. 
To this end, we first derive the stability condition in Section~\ref{sec:stability} and then present the optimal pricing problem in Section~\ref{sec:optimal_pricing}. As per classic results of many-to-many stable matching, the pricing problem is reduced to a linear program whose feasible set is described by the stability condition. Yet, the stability condition is rather difficult to obtain in this study due to the complex assignment. These issues are discussed in detail below, along with the proposed solution approaches.

\subsubsection{Stability condition}\label{sec:stability}
The stability condition states that no travelers and operators can form a new matching that yields a higher total payoff. In our setting, it means no travelers and operators that join the MaaS system can achieve a higher total payoff by shifting to a path that is not used in the optimal assignment of MaaS trips. 
Let $u_{w,l}$ denote the payoff of travelers taking path $l$ between OD $w$ and $v_{m,l}$ be the corresponding payoff of operator $m$. Then, the stability condition is given by 
\begin{align}
    u_{w,l} + \sum_{m\in\M} v_{m,l} &\geq u_{w,l'} + \sum_{m\in\M} v_{m,l'},\; \forall l\in{\setP}^*_w,\; l'\in{\setP}_w\setminus{\setP}^*_w ,\; w\in\W, \label{eq:stable_cond_MaaS}\\
    u_{w,l} + \sum_{m\in\M} v_{m,l} &\geq \tilde{u}_{w,l'} + \sum_{m\in\M} \tilde{v}_{m,l'},\; \forall l\in{\setP}^*_w,\; l'\in{\setP}_w,\; w\in\W, \label{eq:stable_cond_nonMaaS},
\end{align}
where $\setP^*_w$ denote the set of paths taken by MaaS travelers under the optimal assignment solved in Section~\ref{sec:assignment}. 
Eq.~\eqref{eq:stable_cond_MaaS} states that all MaaS trips on matched MaaS paths yield higher total payoffs than MaaS trips on \textit{non-matched} MaaS paths. In other words, travelers and operators have no incentives to still remain in the MaaS system but match on a path that is not used by any other MaaS travelers. 
On the other hand, Eq.~\eqref{eq:stable_cond_nonMaaS} requires all matched MaaS paths to have higher total payoff than all non-MaaS paths. Hence, travelers and operators have no incentive to leave the MaaS system.

Let $U_w$ be the utility value of trips between OD pair $w\in\W$. With a mild ablation of notation, we define $\mu^*_a$ as the dual variable associated with each link at the optimal assignment. Except for the capacitated MT regular links, the value is always zero. And for MT regular links, $\mu^*_a$ can be obtained approximately as the value of $\mu_a^{(k)}$ in Algorithm~\ref{alg:assignment} at convergence.

Then, the net utility for MaaS and non-MaaS travelers choosing $l\in\setP_w$ are given by 
\begin{align}
    u_{w,l} &= U_w - \sum_{a\in l}(t^*_a+\mu^*_a) - p^d_w, \label{eq:MaaS_traveler_payoff}\\
    \tilde{u}_{w,l} & = U_w - \sum_{a\in l}(t^*_a + \tilde{p}_a + \mu^*_a). \label{eq:nonMaaS_traveler_payoff}
\end{align}

Recall that the total service capacity is assumed to be fixed and exogenous. Hence, it is sufficient to use the revenue to define the operator's payoff.
For each operator $m\in\M$, the marginal revenue generated on each path $l$ is given by 
\begin{align}
    v_{m,l} &= \sum_{a\in l\cap\A^m}p^s_a , \label{eq:MaaS_operator_payoff}\\ 
    \tilde{v}_{m,l} &= \sum_{a\in l\cap\A^m} \tilde{p}_a . \label{eq:nonMaaS_operator_payoff}
\end{align}

Proposition~\ref{prop:stable_cond_simp} gives the simplified stability condition based on the optimal assignment derived in Section~\ref{sec:assignment}, while a supporting lemma is first presented below. 

\begin{lemma}[Generalized cost at optimal MaaS assignment]\label{lemma:MaaS_UE_cost}
    Under the optimal MaaS assignment, all active MaaS paths $l\in{\setP}^*_w$ have the same generalized cost $\pi^*_w$:
    \begin{align}
        \pi^*_w = \sum_{a\in l}(t^*_a + \mu^*_a), \forall l\in{\setP}^*_w, w\in\W.
    \end{align} 
\end{lemma}
\begin{proof}
    See Appendix \ref{proof:MaaS_UE_cost}.
\end{proof}

\begin{proposition}[Reduced stability conditions for MaaS pricing]\label{prop:stable_cond_simp}
    Suppose the capacity price is set to be $p^s_a=p^s\lambda_a$, where $\lambda_a>0$ is an exogenous attribute of link $a$. Then, the stability condition can be reduced as follows:
    \begin{enumerate}
        \item The first condition \eqref{eq:stable_cond_MaaS} is naturally satisfied if $p^s_a\geq0,\forall a\in\A_\text{MT}\cup\A_\text{MoD}$. 
        \item The second condition \eqref{eq:stable_cond_nonMaaS} is equivalent to 
        \begin{align}
            p^d_w \leq \min_{l'\in\setP_w} \left\{ \sum_{a\in l'} (t^*_a + \mu^*_a) + \sum_{a\in l'\cap \A_\text{tr}} P_0\right\} - \pi^*_w  + p^s \min_{l\in \setP^*_w} \left\{\sum_{a\in l} \lambda_a\right\}, \label{eq:stable_cond_nonMaaS_simp}
        \end{align}
    \end{enumerate}
\end{proposition}
\begin{proof}[Proof of statement 1.]
    First, note that the capacity purchase price is zero for links without MaaS flows, i.e., $p^s_a = 0$ if $x_a^* = 0$. 
    Plugging Eqs.~\eqref{eq:MaaS_traveler_payoff} and \eqref{eq:MaaS_operator_payoff} into Eq.~\eqref{eq:stable_cond_MaaS}, we have $\forall l\in{\setP}^*_w,\; l'\in{\setP}_w\setminus{\setP}^*_w ,\; w\in\W$, 
    \begin{align}
        &\left(U_w - \sum_{a\in l} (t^*_a + \mu^*_a) - p^d_w \right)+ \left(\sum_{m\in\M} \sum_{a\in l\cap\A^m}p^s_a\right)\geq U_w - \sum_{a\in l'} (t^*_a + \mu^*_a) - p^d_w\nonumber\\
        \Rightarrow \quad &- \sum_{a\in l} t^*_a + \sum_{a\in l}(p^s_a - \mu^*_a) \geq - \sum_{a\in l'} t^*_a - \sum_{a\in l'}\mu^*_a\\
        \Rightarrow \quad & \sum_{a\in l} p^s_a \geq \sum_{a\in l}(t^*_a + \mu^*_a)  - \sum_{a\in l'}(t^*_a + \mu^*_a).
    \end{align}
    As per Lemma~\ref{lemma:MaaS_UE_cost}, we have 
    \begin{align}
        \pi^*_w = \sum_{a\in l}(t^*_a + \mu^*_a)  \leq  \sum_{a\in l'}(t^*_a + \mu^*_a),\quad \forall l\in{\setP}^*_w,\; l'\in{\setP}_w\setminus{\setP}^*_w ,\; w\in\W,
    \end{align}
    where $\pi^*_w$ is the minimum travel time between OD pair $w$ under the optimal assignment. 

    Therefore, the stability condition \eqref{eq:stable_cond_MaaS} holds as long as 
    \begin{align}
        \sum_{a\in l} p^s_a \geq \pi^*_w - \min_{l'\in \setP_w\setminus \setP^*_w} \left\{\sum_{a\in l'}(t^*_a + \mu^*_a)\right\}, \quad \forall w\in \W
    \end{align}
    which is naturally satisfied with the feasible capacity price $p^s_a\geq 0,\forall a\in\A_\text{MT}\cup\A_\text{MoD}$. 
\end{proof}

\begin{proof}[Proof of statement 2.]
    Similar to the proof of statement 1, we first plug Eqs.~\eqref{eq:nonMaaS_traveler_payoff} and \eqref{eq:nonMaaS_operator_payoff} into Eq.~\eqref{eq:stable_cond_nonMaaS} to expand the condition as
    \begin{align}
         &\left(U_w - \sum_{a\in l} (t^*_a + \mu^*_a) - p^d_w \right)+ \left(\sum_{m\in\M} \sum_{a\in l\cap\A^m}p^s_a\right)\\
        &\geq  \left(U_w - \sum_{a\in l'} (t^*_a + \mu^*_a + \tilde{p}_a)\right) + \left(\sum_{m\in\M} \sum_{a\in l'\cap\A^m}\tilde{p}_a \right)\nonumber\\
        \Rightarrow \quad & -\sum_{a\in l} (t^*_a + \mu^*_a - p^s_a) - p^d_w \geq -\sum_{a\in l'} (t^*_a + \mu^*_a) - \sum_{a\in l'\cap \A_\text{tr}} P_0 \\
        \Rightarrow \quad & p^d_w \leq \sum_{a\in l'} (t^*_a + \mu^*_a)  + \sum_{a\in l'\cap \A_\text{tr}} P_0 - \sum_{a\in l} (t^*_a + \mu^*_a  - p^s_a).
    \end{align}
    To further simplify the condition, we replace $\pi^*_w = \sum_{a\in l}(t^*_a+\mu^*_a)$ as per Lemma~\ref{lemma:MaaS_UE_cost} and rewrite $p^s_a = p^s\lambda_a$ as per the proposed pricing scheme. 
    Since the inequality holds for any $l\in\setP^*_w$ and $l'\in{\setP}_w$, it is thus equivalent to 
    \begin{align}
        p^d_w \leq \min_{l'\in{\setP}_w}\left\{\sum_{a\in l'} (t^*_a + \mu^*_a) + \sum_{a\in l'\cap \A_\text{tr}} P_0\right\} - \pi^*_w + p^s \min_{l\in\setP^*_w}\left\{\sum_{a\in l}\lambda_a\right\}.
    \end{align}
\end{proof}
We note that the capacity pricing scheme proposed in Proposition~\ref{prop:stable_cond_simp} effectively avoids the challenge of enumerating the path set under optimal MaaS assignment $\setP^*_w$ and thus makes the solution approach scalable.

The stability condition can be further simplified if the following sufficient condition is satisfied. 
\begin{corollary}[Sufficient stability condition]
If $p^s \geq 0$ and $p^d_w \leq U_w - \pi^*_w,\forall w\in W$, then a sufficient stability condition reads
\begin{align}
    \tau^*_{w,\text{min}}= \min_{l'\in\setP_w} \left\{ \sum_{a\in l'} (t^*_a + \mu^*_a) + \sum_{a\in l'\cap \A_\text{tr}} P_0\right\} \geq U_w, \; \forall w \in \mathcal{W}.
\end{align}
\end{corollary}
\begin{proof}
    If $\tau^*_{w,\text{min}} \geq U_w$, we have:
    \begin{subequations}
        \begin{align}
        p^d_w &\leq U_w - \pi^*_w \leq \tau^*_{w,\text{min}} - \pi^*_w \leq \tau^*_{w,\text{min}} - \pi^*_w + p^s \lambda^*_\text{min}, \nonumber
        \end{align}
    \end{subequations}
    where $\lambda^*_{w, \text{min}} = \min_{l\in \setP^*_w} \left\{\sum_{a\in l} \lambda_a\right\}$. 
    Hence, the stability condition~\eqref{eq:stable_cond_nonMaaS_simp} is satisfied. 
\end{proof}

\subsubsection{Optimal pricing}\label{sec:optimal_pricing}
Now we are ready to present the optimal pricing problem as follows:
\begin{subequations}\label{eq:pricing}
   \begin{align}
        \max_{p^d,p^s}\quad & \sum_{w\in\W} p^d_w q^*_w - \sum_{a\in \A_r} p^s_a x^*_a,\\
        s.t.\quad & u_{w,l}\geq 0,\quad && \forall l\in{\setP}^*_w, w\in\W,\label{eq:pricing_constraint_traveler_payoff}\\
        & \sum_{a\in\A^m_\text{MT-r}} (p^s_ax^*_a + \tilde{p}_a\tilde{x}^*_a) \geq B_m,\quad && \forall m\in{\M}_\text{MT},\label{eq:pricing_constraint_MToperator_revenue}\\
        & \sum_{a\in\A^m_\text{MoD-r}} (p^s_ax^*_a + \tilde{p}_a\tilde{x}^*_a) \geq B_m,\quad && \forall m\in{\M}_\text{MoD},\label{eq:pricing_constraint_MoDoperator_revenue}\\
        & \text{Stability condition \eqref{eq:stable_cond_MaaS} and \eqref{eq:stable_cond_nonMaaS}}.
    \end{align} 
\end{subequations}
Constraint~\eqref{eq:pricing_constraint_traveler_payoff} requires that all MaaS trips yield a non-negative payoff and it is equivalent to
\begin{align}
    p^d_w \leq U_w - \sum_{a\in l} (t^*_a + \mu^*_a) = U_w - \pi^*_w.
\end{align}
Constraints~\eqref{eq:pricing_constraint_MToperator_revenue} and \eqref{eq:pricing_constraint_MoDoperator_revenue} state the total revenue of each operator $m$ must be no less than $B_m$, which can be specified as the operator's revenue prior to the introduction of MaaS platform.

Due to the proposed capacity purchase pricing scheme and Proposition~\ref{prop:stable_cond_simp}, the pricing problem can be simplified as 
\begin{subequations}\label{eq:pricing_simp}
   \begin{align}
        \max_{p^d,p^s}\quad & \sum_w p^d_w q^*_w - p^s\sum_{a\in \A_r} \lambda_a x^*_a,\\
        s.t.\quad & p^d_w \leq U_w - \pi^*_w,\quad && \forall 
        w\in\W, \label{eq:pricing_constraint_MaaS_traveler_nonneg_payoff_explicit}\\
        & p^d_w \leq \tau^*_{w,\text{min}} - \pi^*_w  + p^s \lambda^*_{w, \text{min}},\quad && \forall w\in\W, \label{eq:pricing_constraint_MaaS_traveler_stable_explicit}\\
        & \sum_{a\in\A^m_\text{MT-r}} (p^s\lambda_ax^*_a + \tilde{p}_a\tilde{x}^*_a) \geq B_m,\quad && \forall m\in{\M}_\text{MT}, \label{eq:pricing_constraint_MToperator_revenue_factorized}\\
        & \sum_{a\in\A^m_\text{MoD-r}} (p^s\lambda_ax^*_a + \tilde{p}_a\tilde{x}^*_a) \geq B_m,\quad && \forall m\in{\M}_\text{MoD},\label{eq:pricing_constraint_MoDoperator_revenue_factorized}\\
        & p^s \geq 0\label{eq:pricing_constraint_nonneg},
    \end{align} 
\end{subequations}
where $\tau^{*}_{w,\text{min}} = \min_{l'\in\setP_w} \left\{ \sum_{a\in l'} (t^*_a + \mu^*_a) + \sum_{a\in l'\cap \A_\text{tr}} P_0\right\}$ and $\lambda^*_{w, \text{min}} = \min_{l\in \setP^*_w} \left\{\sum_{a\in l} \lambda_a\right\}$ are readily computed from the optimal assignment results.

Consequently, problem~\eqref{eq:pricing_simp} reduces to a simple linear program and can be solved via any commercial solver. Moreover, the proposed pricing problem guarantees the existence of a stable outcome and often a finite unique optimal pricing scheme given each optimal assignment. These results are formally presented below.

\begin{proposition}[Existence of stable outcome and finite optimal pricing]\label{prop:existence_stable_outcome}
    The stable outcome space of the MaaS many-to-many matching problem is nonempty. Further, given an optimal assignment $(q^*, x^*)$, the pricing problem \eqref{eq:pricing_simp} has at least one finite optinoptimal solution.
\end{proposition}

\begin{proof}
    We first prove the stable outcome space is nonempty. Recall that Proposition~\ref{prop:existence_assignment} proves the existence of optimal assignment. Hence, the remaining task is to prove, given an optimal assignment $(q^*, x^*)$, the pricing problem \eqref{eq:pricing_simp} is feasible, i.e., there exists a pricing scheme $p$ that satisfies Constraints~\eqref{eq:pricing_constraint_MaaS_traveler_nonneg_payoff_explicit}-\eqref{eq:pricing_constraint_nonneg}. This is trivial because an infinite solution $p_w^d = -\infty,\forall w\in\W$ and $p^s = \infty$ is feasible.

    Since the pricing problem is feasible, it always has optimal solutions. We next show the optimal solution is not taken at infinity. Note that the feasible set of \eqref{eq:pricing_simp} is closed. As per the property of linear programs, the optimal solution is always obtained at the boundary of the feasible set. 
    Now consider a boundary solution such that Constraint~\eqref{eq:pricing_constraint_MaaS_traveler_nonneg_payoff_explicit} is binding. Then, we have 
    \begin{align}
        p_w^{d} = U_w - \pi_w^*,\quad \forall w\in\W.\label{eq:binding_p_d}
    \end{align}
    Accordingly, Constraint~\eqref{eq:pricing_constraint_MaaS_traveler_stable_explicit} reduces to 
    \begin{align}
        p^s \geq \frac{U_w - \tau^*_{w,\text{min}}}{\lambda^*_{w, \text{min}}},\quad \forall w\in\W.
    \end{align}
    Further combined with Constraints~\eqref{eq:pricing_constraint_MToperator_revenue_factorized}-\eqref{eq:pricing_constraint_nonneg} yields
    \begin{align}\label{eq:binding_p_s}
        p^{s} \geq  \max\left\{0, \max_w \frac{U_w - \tau^*_{w,\text{min}}}{\lambda^*_{w, \text{min}}}, \max_{m\in \M} \frac{B_m - \sum_{a\in\A^m_\text{r}}\tilde{p}_a\tilde{x}^*_a}{\sum_{a\in\A^m_\text{r}} \lambda_ax^*_a} \right\}
    \end{align}
    Clearly, the objective value under the solution $p^d_w$ as per Eq.~\eqref{eq:binding_p_d} and $p^s$ as the lower bound of Eq.~\eqref{eq:binding_p_s} is higher than that under the infinite solution. Hence, the optimal solution must be finite. 
\end{proof}

\begin{corollary}[Uniqueness of optimal pricing]\label{coro:unique_optimal_pricing}
    Given a positive optimal assignment $(q^*,x^*)$ ($\sum_w q^*_w >0$), there exists a unique optimal pricing $p^{s*}$ and $p^{d*}_w,\;\forall w\in\W\; s.t.\; q^*_w>0$ under the following conditions:
    \begin{enumerate}
        \item There exists $w\in\W$ such that $q^*_w \lambda_{w, \text{min}}^*\neq \frac{1}{|\mathcal{W}|} \sum_{a\in \A} \lambda_a x^*_a$;
        \item There exists $w, w' \in \W$ such that $\frac{U_w - \tau^*_{w,\text{min}}}{\lambda^*_{w, \text{min}}} \neq \frac{U_{w'} - \tau^*_{{w'},\text{min}}}{\lambda^*_{{w'}, \text{min}}}$.
    \end{enumerate}
\end{corollary}
\begin{proof}
    To prove the solution uniqueness, we introduce an auxiliary price $p^s_w$ and rewrite the pricing problem as follows:
    \begin{subequations}\label{eq:pricing_aux}
       \begin{align}
            \max_{p^d,p^s}\quad & \sum_w \left(p^d_w q^*_w - \frac{p^s_w}{|\W|}\sum_{a\in \A_r} \lambda_a x^*_a\right),\\
            s.t. \quad & p^s_w = p^s_{w'}, \quad && \forall w,w'\in\W,\label{eq:pricing_aux_consensus}\\
            & p^d_w \leq U_w - \pi^*_w,\quad && \forall w\in\W,\\
            & p^d_w \leq \tau^*_{w,\text{min}} - \pi^*_w  + p^s_w \lambda^*_{w, \text{min}},\quad && \forall w\in\W,\\
            & \sum_{a\in\A^m_\text{r}} (p^s_w\lambda_ax^*_a + \tilde{p}_a\tilde{x}^*_a) \geq B_m,\quad && \forall w\in\W, m\in{\M}, \\
            & p^s_w \geq 0.
        \end{align} 
    \end{subequations}
    Problem \eqref{eq:pricing_aux} can be seen as a distributed optimization problem~\citep{yang2019survey}, where each OD pair solves its own price $(p^d_w, p^s_w)$ and ensures a consensus constraint Eq.~\eqref{eq:pricing_aux_consensus}. Specifically, each subproblem is given by 
    \begin{subequations}\label{eq:pricing_local}
        \begin{align}
            \max_{p^d_w,p^s_w}\quad & p^d_w q^*_w - \frac{p^s_w}{|\W|}\sum_{a\in \A_r} \lambda_a x^*_a,\\
            s.t. \quad & p^d_w \leq U_w - \pi^*_w,\\
            & p^d_w \leq \tau^*_{w,\text{min}} - \pi^*_w  + p^s_w \lambda^*_{w, \text{min}},\label{eq:pricing_local_p_d}\\
            & p^s_w \geq \underline{p}^s = \max \left\{0, \max_{m\in \M} \frac{B_m - \sum_{a\in\A^m_\text{r}}\tilde{p}_a\tilde{x}^*_a}{\sum_{a\in\A^m_\text{r}} \lambda_ax^*_a}\right\}. 
        \end{align}
    \end{subequations}
    A graphical illustration of subproblem~\eqref{eq:pricing_local} is presented in Figure~\ref{fig:optimal_pricing_region}. 
    \begin{figure}[htb]
    \centering
    \includegraphics[width=0.85\textwidth]{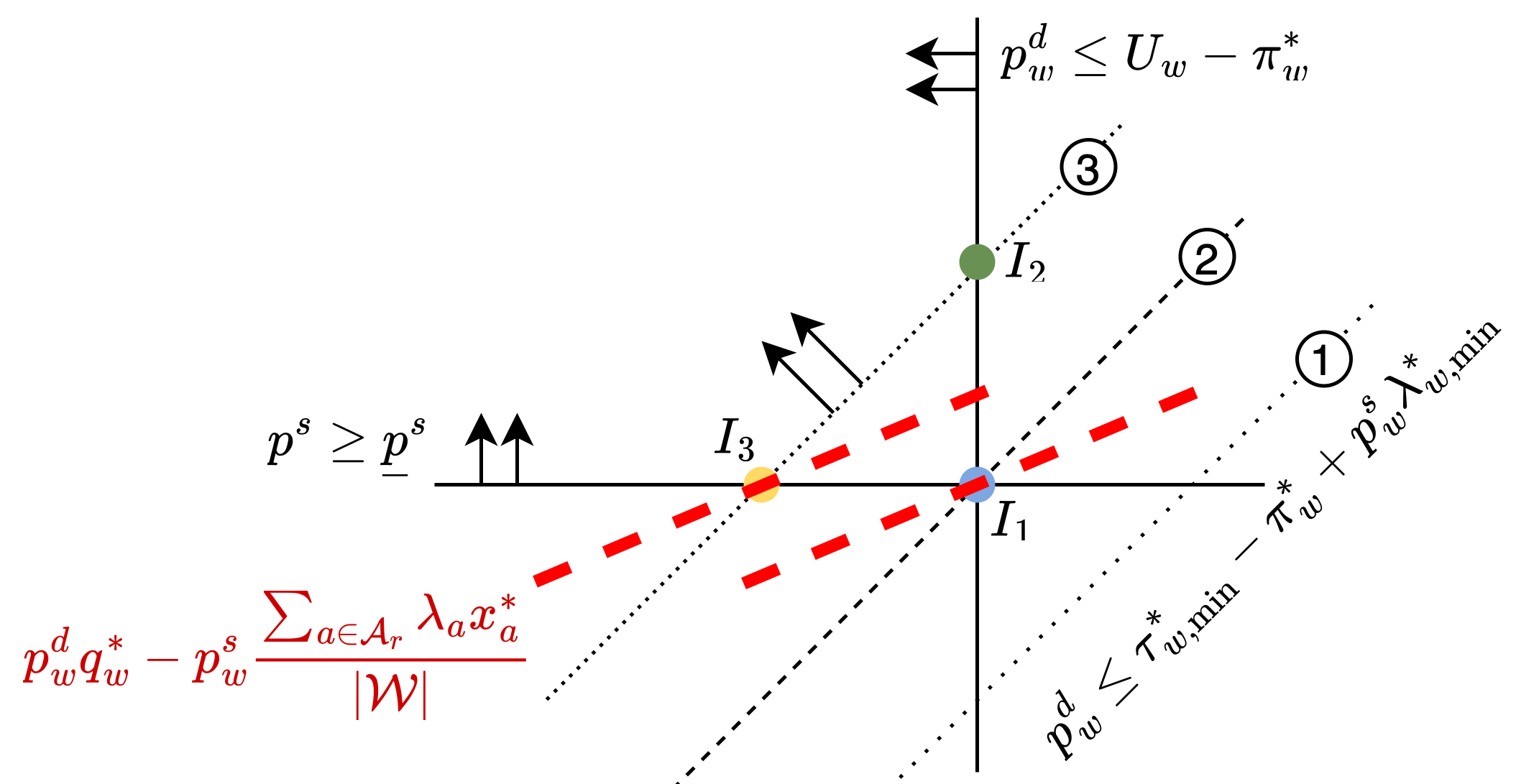}
    \caption{Graphical illustration of subproblem $w$.}
    \label{fig:optimal_pricing_region}
    \end{figure}
    Based on the location of Constraint~\eqref{eq:pricing_local_p_d}, we end up with different scenarios of optimal solutions. 
    For Cases 1 and 2 in Figure~\ref{fig:optimal_pricing_region}, there exists only one extreme point
    \begin{itemize}
        \item $I_1$: $p^d_w = U_w - \pi^*_w$ and $p^s_w = \underline{p}^s$.
    \end{itemize} 
    It is thus the unique optimal solution under the condition $q^*_w\neq 0$ and $\sum_{a\in\A_r}\lambda_a x^*_a \neq 0$. The latter naturally holds when the total MaaS demand is positive ($\sum_w q^*_w >0$). 
    
    In Case 3, there are two extreme points:
    \begin{itemize}
        \item $I_2$: $p^d_w = U_w - \pi^*_w$ and $p^s_w = \frac{U_w - \tau^*_{w,\text{min}}}{\lambda^*_{w, \text{min}}}$; 
        \item $I_3$: $p^d_w = \tau^*_{w,\min} - \pi^*_w + \underline{p}^s\lambda^*_{w,\min}$ and $p^s_w = \underline{p}^s$.
    \end{itemize}
    Clearly, $p^s_w$ in $I_2$ is unlikely to satisfy the consensus because it is determined by some OD-specific variables. Formally, if there exist  $w,w'\in\W$ such that $\frac{U_w - \tau^*_{w,\text{min}}}{\lambda^*_{w, \text{min}}} \neq \frac{U_{w'} - \tau^*_{w',\text{min}}}{\lambda^*_{w', \text{min}}}$, then $p^{s*}$ cannot be reached at a $I_2$-type extreme point. This gives the second condition in Corollary~\ref{coro:unique_optimal_pricing}. 

    In contrast, $p^s_w$ in $I_3$ does not vary among OD pairs. To ensure it is the unique optimal solution for the subproblem, we must have $q^*_w \neq 0$ and Constraint~\eqref{eq:pricing_local_p_d} is not parallel with the objective function (see Figure~\ref{fig:optimal_pricing_region}). The latter is equivalent to the condition $q^*_w\lambda^*_{w,\min} \neq \frac{1}{|\W|}\sum_{a\in \A_r} \lambda_a x^*_a$. Note that we only need this condition to hold for at least one OD pair (one subproblem has unique optimal $p^s_w = \underline{p^s}$), which gives the first condition in Corollary~\ref{coro:unique_optimal_pricing}. 
\end{proof}

\section{Numerical experiments}\label{sec:results}

\subsection{Experiment setup}\label{subsec:results_setup}
In this section, we conduct numerical experiments on a multi-modal mobility network extended from the Sioux Falls network. As illustrated in Figure~\ref{fig:SF_MaaS_net}, we add bidirectional metro lines and bus lines on top of the existing road network. In all experiments, we consider a single MoD platform providing e-hailing services. This introduces another copy of the road network, along with the access links. although the travel time on each link in the MoD subnetwork is jointly determined by the total flow that includes private vehicles. 
The subnetworks are then connected with each other through transfer links. The resulting network has 199 nodes and 456 links, compared to 24 nodes and 76 links in the original one. The number of OD pairs also doubles (1056 pairs) as we consider both MaaS and non-MaaS travelers. 
The link parameters and other exogenous variables are summarized in Appendix~\ref{sec:SF_link_parameter}. 
Note that we do not specify the value of time in the model and thus all fare costs are measured in equivalent travel time.

\begin{figure}[htb]
    \centering
    \includegraphics[width=0.8\textwidth]{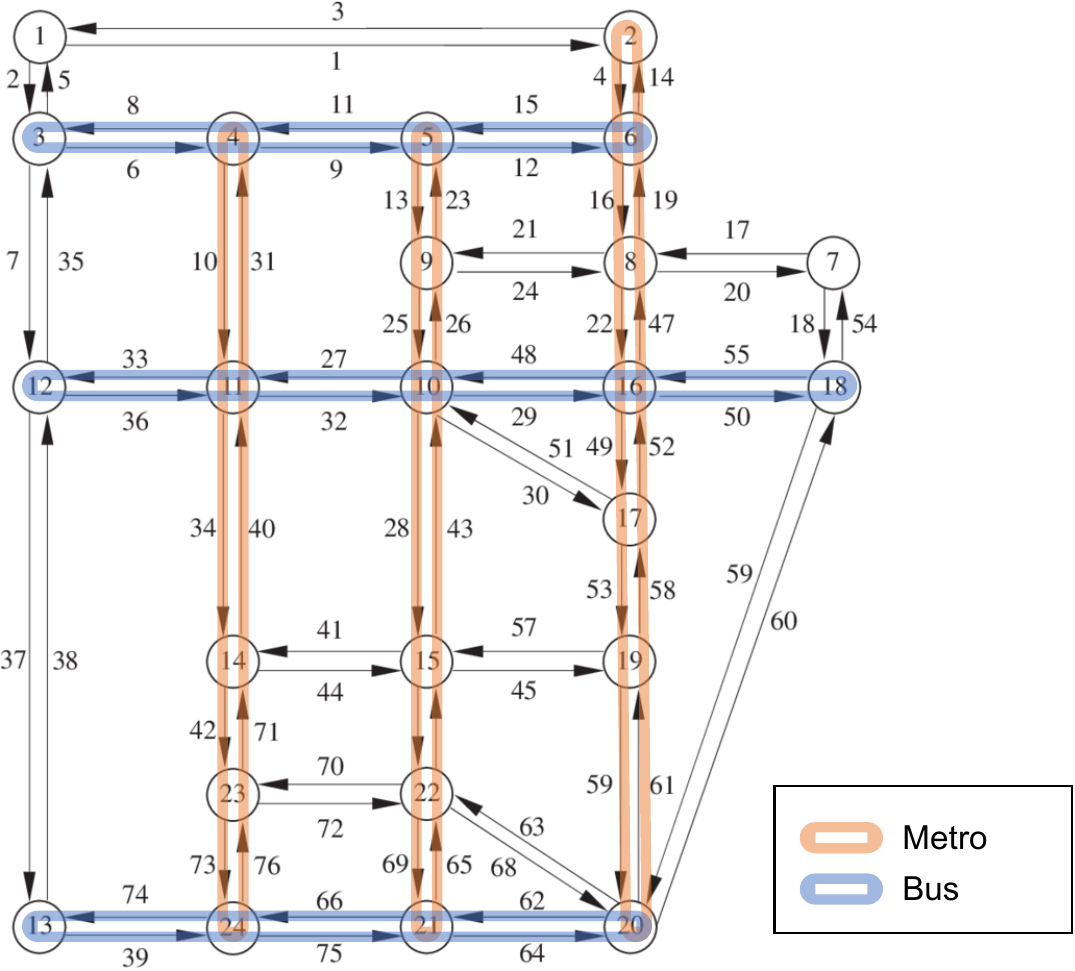}
    \caption{Multi-modal mobility network for Sioux Falls.}
    \label{fig:SF_MaaS_net}
\end{figure}

\begin{figure}[htb]
    \centering
    \includegraphics[width=0.65\textwidth]{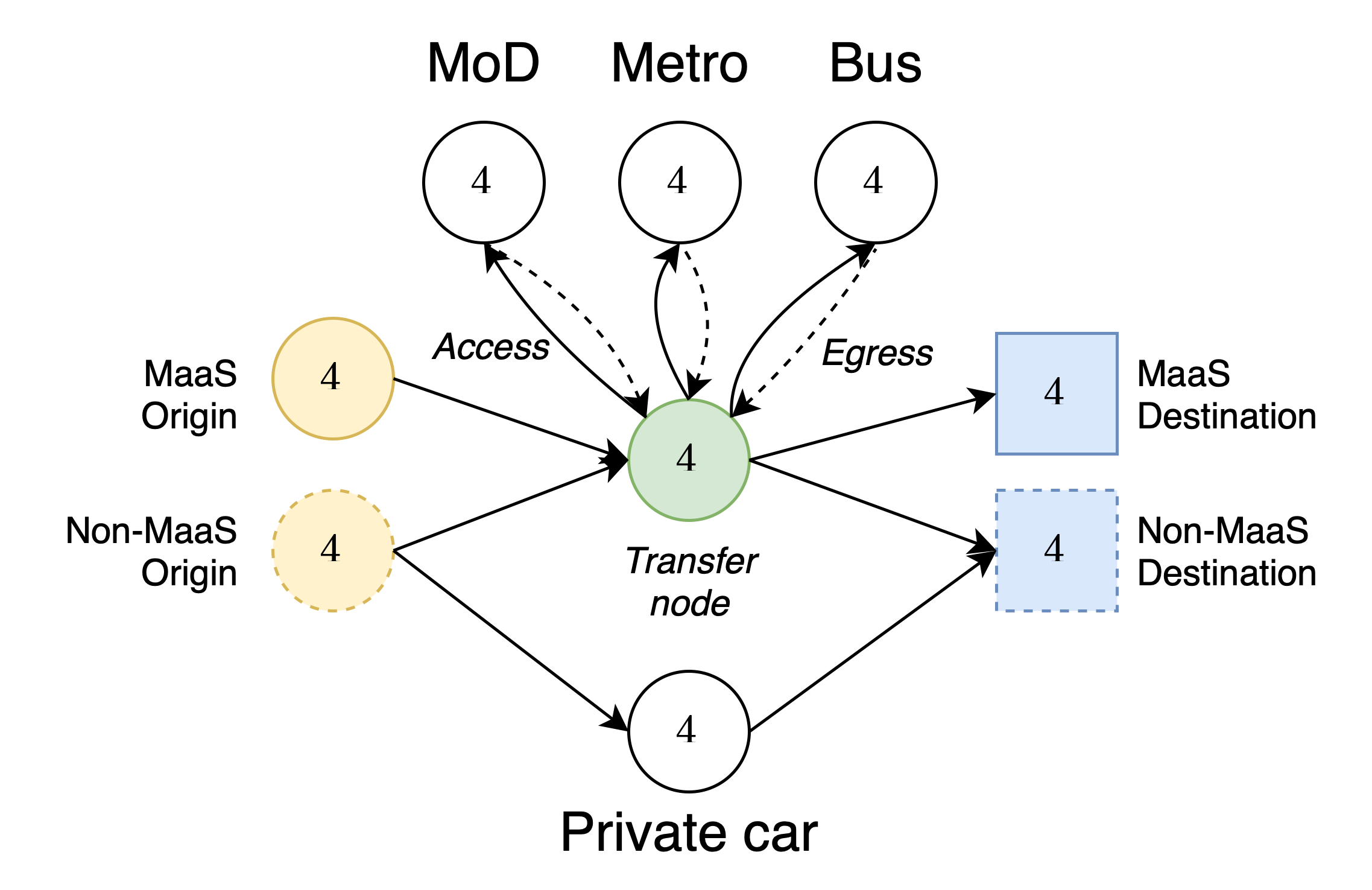}
    \caption{Example of multi-modal network connections.}
    \label{fig:transfer_node_example}
\end{figure}

Figure~\ref{fig:transfer_node_example} presents the multi-modal connections at node 4. To distinguish different types of traffic flows, we create nine copies of the node. Four of them refer to the origin and destination of MaaS and non-MaaS trips, respectively. They are connected to the transfer node through dummy links, and the transfer node is then connected to the MT and MoD subnetworks by access and egress links. Since non-MaaS travelers may choose to drive, the non-MaaS origin and destination nodes are also connected to node 4 in the road network. 
To simplify the network, we further merge transfer links with access links. This is equivalent to adding the transfer time on each access link meanwhile subtracting it on the dummy link from each MaaS origin node to the corresponding transfer node. The additional planning cost $P_0$ in non-MaaS trips is handled in a similar way.

In Proposition~\ref{prop:stable_cond_simp}, we introduce the capacity pricing scheme $p^s_a = p^s\lambda_a$, where $\lambda_a$ denotes some exogenous attribute of a regular link $a$. In the experiments, we set it as follows:
\begin{align}\label{eq:lambda}
    \lambda_a = 
    \begin{cases}
        \tilde{p}_a, & a\in\A_\text{MoD-r}\\
        \eta \tilde{p}_a, & a\in\A_\text{MT-r}\\
        0.5 t_a, & a\in\A_\text{MT-a}\\
        0, & \text{otherwise}
    \end{cases}.
\end{align}
Recall that $\tilde{p}_a$ is the non-MaaS trip fare on link $a$. Eq.~\eqref{eq:lambda} indicates that the MaaS platform can purchase capacity at a wholesale price and the discount is given by $1-p^s$. Considering the platform may have different negotiation power with MT and MoD operators, we further introduce $\eta$ as an exogenous MT capacity price factor. Lastly, we assume the MaaS also compensates for half of the operation cost in the MoD pickup process. Here, $t_a$ is considered exogenous because it is the output of the assignment problem.

All the numerical experiments are performed on an HPC cluster with a V100-32GB GPU.
The MaaS assignment problem is initialized from the base scenario without the MaaS platform (i.e., $q_0 = 0, x_0 = 0$, and $\tilde{x}_0$ as the corresponding equilibrium flows), such that the setting mimics the process of introducing MaaS into an existing transportation system. 
With this in consideration, we also set the trip utility $U_w$ equal to the equilibrium cost in the base scenario, i.e., travelers get zero net utility when there is no MaaS. The total runtime for the MaaS assignment problem is 3.68 hours, and the convergence results are shown in Figures~\ref{fig:outer_convergence} and \ref{fig:inner_convergence} in Appendix~\ref{sec:algo_performance}.

\subsection{Assignment results}

The main results of the assignment problem are reported in Table~\ref{tab:assignment_results}, along with those in the base scenario without the MaaS platform. 

Table~\ref{tab:assignment_results} first reports the split between MaaS and non-MaaS trips. Under the optimal assignment, 55.81\% of travelers would choose MaaS. Most of them use to take public transport and plan their multi-modal trips on their own, though there is also a considerable amount of demand that switch from driving. This demonstrates the potential of MaaS for substituting private driving with public transport by reducing the inconvenience of trip planning and fare payment. Accordingly, the share of demand for public transport increases from 61.40\% to 74.38\%. 

Although the assignment results do not produce the path flow, we could use the traffic flows traversing the transfer node to compute the aggregate mode transfer. As shown in Table~\ref{tab:assignment_results}, before the launch of MaaS, the average number of mode transfers in the multi-modal mobility network is 0.16. If each trip has at most one transfer, which is a reasonable assumption given the additional transfer time and planning cost, it implies 16\% of the trips are multi-modal. The value increases to 0.18 after MaaS is introduced, and the result suggests MaaS travelers take more multi-modal trips, compare to non-MaaS travelers. Some of these trips may already exist in the base scenario as the mode transfer in non-MaaS trips decreases. However, the overall 8.04\% growth in mode transfer indicates MaaS indeed promotes multi-modal travel. 

\begin{table}[H]
\caption{Assignment results}
\label{tab:assignment_results}
\centering
\setlength\tabcolsep{8pt}
\begin{threeparttable}
\begin{tabular}{lcl}
\toprule
 & \textbf{Without platform} & \textbf{With platform} \\
\midrule
\multicolumn{3}{l}{{\textbf{Travel demand}}}                                              \\
{Non-MaaS}                         & 100\%                         & 44.19\%                     \\
\textit{\quad - MT \& MoD}             & 61.40\%                       & 18.57\%  \quad(-42.83\%)                    \\
{\textit{\quad - Private driving}}     & {38.60\%}      & 25.62\% \quad(-12.98\%)                     \\
\multicolumn{3}{c}{}                                                                                           \\
{MaaS}       & - & 55.81\%                                                                \\
\textit{\quad - MT \& MoD}             & -                             & 55.81\%                     \\
\midrule
\multicolumn{3}{l}{{\textbf{Multi-modal trips}}}\\
{Mode transfer per trip}     & {0.16}      & 0.18 \quad\quad\textcolor{black}{(+8.04\%)}                    \\
\textit{\quad - Non-MaaS}  & 0.16  & 0.11 \quad\quad\textcolor{black}{(-32.99\%)}\\
\textit{\quad - MaaS}  & - & 0.23 \\
\midrule
\multicolumn{3}{l}{{\textbf{Capacity utilization}}}                                                     \\
{MT}    & {44.97\%}      & 48.49\% \quad(+3.52\%)                  \\
\textit{\quad - Non-MaaS}  & 44.97\%  & 17.38\% \quad(-27.59\%)\\
\textit{\quad - MaaS}  & - & 31.11\% \\
{MoD}           & {76.17\%}    & {77.70\%} \quad(+1.53\%)                 \\
\textit{\quad - Non-MaaS}  & 76.17\%  & 11.06\% \quad(-65.11\%)\\
\textit{\quad - MaaS}  & - & 66.64\% \\
\midrule
{{\textbf{Social cost}}} & & \\
Travel time per trip (min)      & 19.02     &    16.71\quad\quad(-12.15\%) \\
Social cost per trip   & 22.83       &  20.53\quad\quad(-10.07\%)\\
\bottomrule
\end{tabular}
\end{threeparttable}
\end{table}

As expected, MaaS encourages ridership of MT and MoD. 
As reported in Table~\ref{tab:assignment_results}, the increase in capacity utilization rate is more significant for MT (3.52\%) compared to MoD (1.53\%), which is contributed by MaaS travelers (contrasting to reduction of 27.59\% and 65.11\% in non-MaaS MT and MoD ridership). This implies the MaaS platform tends to assign more traffic flows on the non-congestible MT network. 
Consequently, the average travel time in the whole system reduces by 12.15\%, and the social cost per trip, i.e., average travel time plus the operation cost equally apportioned to each trip, reduces by 10.07\%.

To better understand how MaaS helps reduce road traffic, we identify pairs of nodes that are connected by both road and transit links (e.g., (3, 4)) and then divide them into two groups based on the volume-over-capacity (VOC) of road links in the scenario without MaaS. The links with VOC $\geq 1$ are referred to as ``congested'' links while those with VOC~$<1$ are named ``non-congested''. Table~\ref{tab:congested_link_flows} compares the total flows, average and maximum VOC of these two groups of links in the scenarios with and without MaaS (for MT links, the VOC is equivalent to the capacity utilization rate). 
It can be seen that MaaS brings a more significant reduction in vehicular traffic on congested links (6.50\%) than non-congested links (1.79\%). 
Accordingly, the average VOC drops by almost 45\% to 0.76, which implies these links are no longer congested after MaaS is launched. Meanwhile, the average VOC also reduces by 31\% for non-congested links. The reduced road traffic flows are distributed to the MT network. Since MaaS also relocates other demand flows to MT, the total increase in MT flows is larger than the total decrease in car flows.

\begin{table}[H]
\caption{Traffic flows on the overlapping road and MT links.}
\label{tab:congested_link_flows}
\centering
\setlength\tabcolsep{10pt}
\begin{threeparttable}
\begin{tabular}{cccccc}
    \toprule
\multirowcell{2}{\textbf{Type of link}} &  & \multicolumn{2}{l}{\textbf{Without platform}} & \multicolumn{2}{l}{\textbf{With platform}} \\
 & & \multicolumn{1}{c}{Car} & \multicolumn{1}{c}{MT} & \multicolumn{1}{c}{Car} & \multicolumn{1}{c}{MT} \\
\midrule
Congested & total flow & 202,951.11& 249,110.41  & 189,754.06 & 269,798.38 \\
& & &  & (-6.50\%) & (+8.30\%)\\
& avg. VOC &  1.38& 0.52&  0.76& 0.56\\
& & & & (-44.93\%) & (+7.69\%) \\
& max VOC &  1.86&  1.00&  1.37& 1.00\\
& & & & (-26.34\%) & (0.00\%) \\
\midrule
Non-congested & total flow & 117,438.83& 88,130.55&  115,331.03& 93,856.80\\
& & &  & (-1.79\%) & (+6.50\%)\\
& avg. VOC &  0.84&  0.33&  0.58& 0.35\\
& & & & (-30.95\%) & (+6.06\%) \\
& max VOC & 0.98&  0.75&  0.91& 0.71\\
& & & & (-7.14\%) & (-5.33\%) \\
\bottomrule   
\end{tabular}
\end{threeparttable}
\end{table}

\subsection{Pricing results}
This section presents the main results of pricing given the optimal MaaS assignment solved in the previous section. As per Corollary~\ref{prop:existence_stable_outcome}, a pricing scheme that satisfies the stability condition is guaranteed to exist. However, as we do not specify the feasible range of MaaS trip fare $p^d$ and its relationship to the capacity purchase price $p^s$, we may end up compensating some MaaS trips (i.e.,$p^w<0$) or an overall negative platform profit. 

As illustrated in Figure~\ref{fig:MaaS_profit_to_MT_factor}, the MaaS platform profit is quite sensitive to the MT capacity price factor $\eta$. Specifically, when $\eta < 0.71$ or $\eta > 1.15$, the platform would suffer from a deficit. The maximum profit is achieved at $\eta = 0.87$ with $p^s = 0.76$, which means the platform pays unit MoD occupied time at a discount of 24\% and unit MT capacity at a discount of 34\%.

Another interesting observation in Figure~\ref{fig:MaaS_profit_to_MT_factor} is that the optimal capacity price does not change anymore as $\eta$ increases beyond 0.75. We note that this is due to the stability condition properties in this particular case study, which will be further discussed at the end of this section.

Table~\ref{tab:optimal_MaaS_OD_pricing} summarizes the main results of optimal MaaS trip fare at $\eta=0.87$. Recall the MaaS platform charges travelers solely based on their origins and destinations, regardless of modes. Hence, the statistics are computed among OD pairs. Specifically, we distinguish those with a negative value ($p^d_w<0$), meaning trips between these OD pairs are actually compensated by the platform. 
In Appendix~\ref{sec:compensated_OD_pair}, we further investigate the compensated OD pairs. In brief, we find all of them are short-distance, i.e., connected by one link, but lack of direct MT service. The MaaS platform tends to serve travelers between these OD pair with MoD service because it benefits the whole system. The platform has to compensate these travelers otherwise they would choose driving because it is convenient and not costly. 

\begin{figure}[H]
\captionsetup[subfigure]{justification=centering}
\centering
{\includegraphics[width=0.8\linewidth]{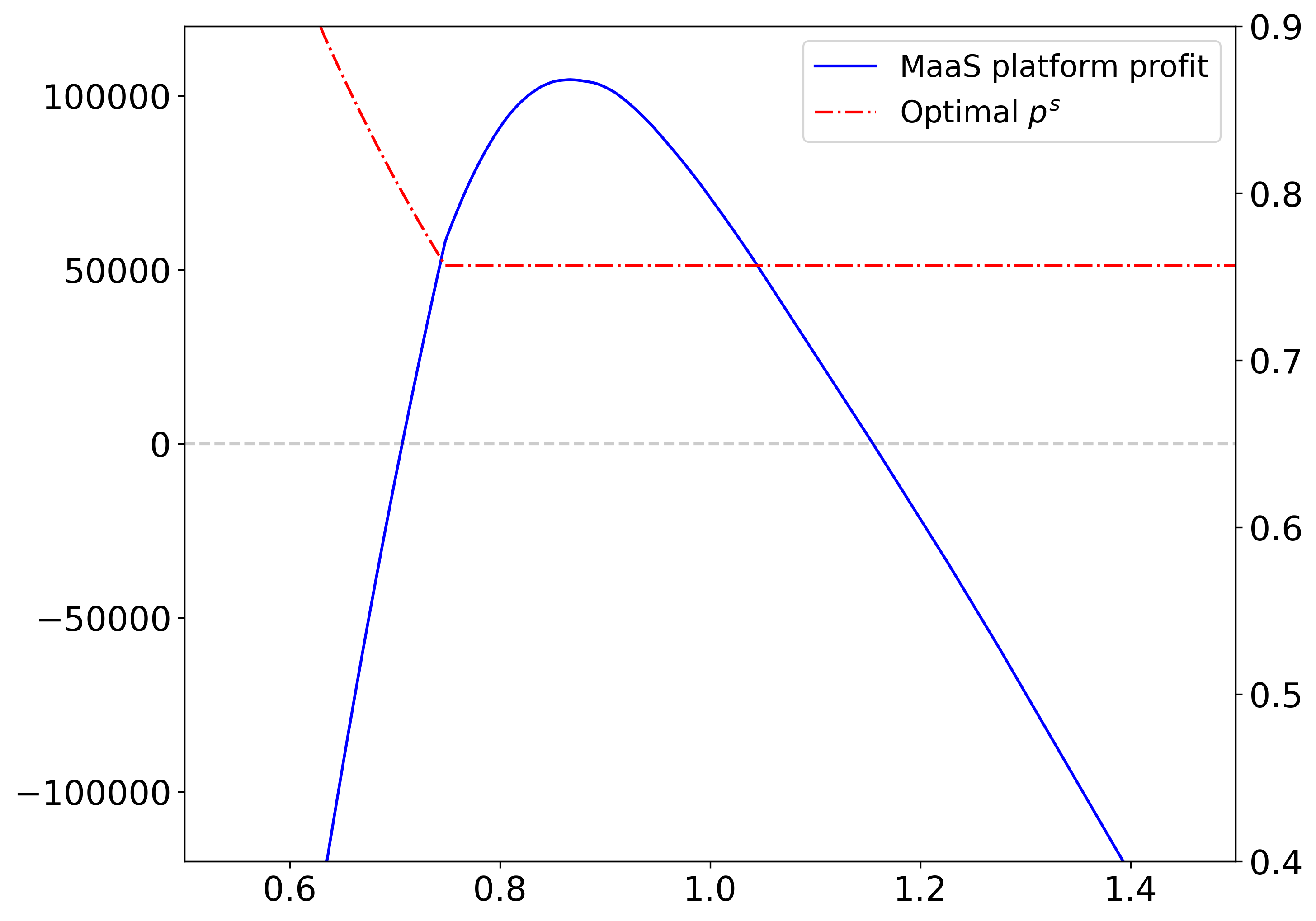}}
\caption{MaaS platform profit with for different MT capacity price factor}
\label{fig:MaaS_profit_to_MT_factor}
\end{figure}

\begin{table}[H]
\caption{Optimal OD-based MaaS pricing}\label{tab:optimal_MaaS_OD_pricing}
\label{tab:price}
\centering
\setlength\tabcolsep{30pt}
\begin{threeparttable}
\begin{tabular}{lccc}
  \toprule
  & \textbf{Min} & \textbf{Avg.} & \textbf{Max}\\
\midrule
  \multicolumn{4}{l}{\textbf{Optimal} $p^d$ } \\
  Fare  & $0.01$ & $10.11$ & $24.63$ \\
  
  Compensation  & $(0.08)$ & $(1.28)$ & $(3.08)$ \\
\midrule
 \multicolumn{4}{l}{\textbf{Fare-to-time ratio}}\\
  MaaS$^*$  & $0.01$ & $0.30$ & $0.44$ \\
  Non-MaaS  & $0.13$ & $0.37$ & $0.60$ \\
\bottomrule   
\end{tabular}
\begin{tablenotes}
\item[*]\textit{Only positive trip fares are considered.}
\end{tablenotes}
\end{threeparttable}
\end{table}

Table~\ref{tab:optimal_MaaS_OD_pricing} also reports the ratio between trip fare and travel time, namely, the fare-to-time ratio, of MaaS and non-MaaS trips, which can be interpreted as the trip fare per unit travel time. It can be found that the MaaS reduces the per-minute fare and particularly, both the lower and upper bounds drop significantly compare to non-MaaS trips.

\begin{figure}[H]
\captionsetup[subfigure]{justification=centering}
\centering
{\includegraphics[width=1\linewidth]{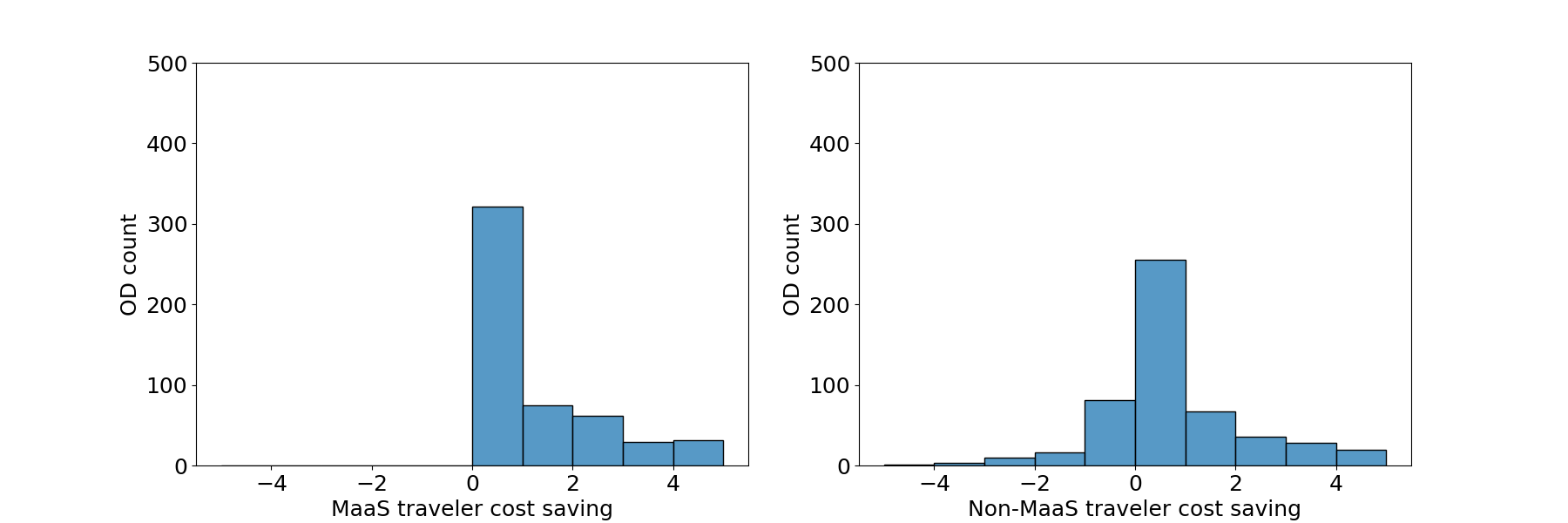}}
\caption{Traveler cost saving (compared to without MaaS platform)}\label{fig:traveler_cost_saving_contrast_withoutMaaS}
\end{figure}

In Figures \ref{fig:traveler_cost_saving_contrast_withoutMaaS}, we compare the total travel cost (the sum of travel time plus trip fare) of MaaS and non-MaaS travelers to the scenario without MaaS by OD pairs. Since we set the trip utility equal to the total travel cost in the base scenario, all MaaS travelers enjoy a non-negative cost saving due to Constraint~\eqref{eq:pricing_constraint_MaaS_traveler_nonneg_payoff_explicit}, though the saving between most OD pairs is limited. This is expected as the MaaS platform aims to maximize its profit when designing the pricing scheme.
In contrast, compared to the scenario without MaaS, some non-MaaS travelers suffer from a larger total travel cost.

\begin{figure}[H]
\captionsetup[subfigure]{justification=centering}
\centering
{\includegraphics[width=0.5\linewidth]{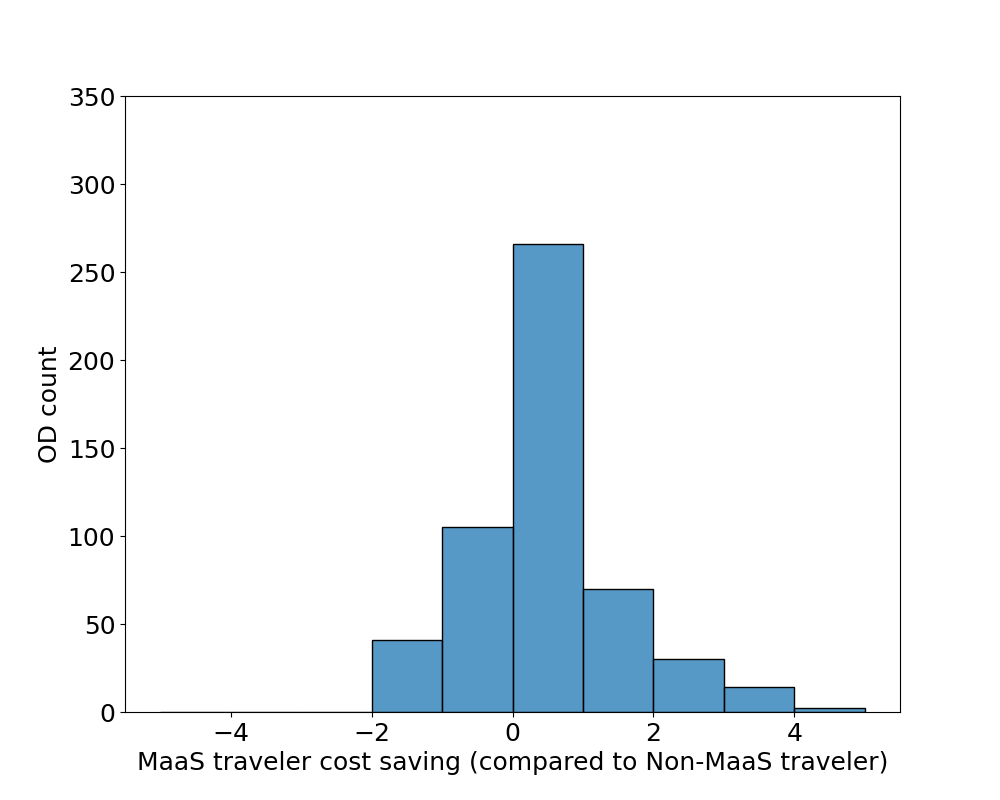}}
\caption{MaaS traveler cost saving (compared to Non-MaaS traveler)}\label{fig:traveler_cost_saving_contrast_NonMaaS}
\end{figure}

Figure \ref{fig:traveler_cost_saving_contrast_NonMaaS} further compares the total travel cost between MaaS and non-MaaS travelers by OD pairs. Interestingly, MaaS travelers do not always enjoy a lower cost than non-MaaS travelers between the same OD pair. We note this phenomenon is due to the condition of stable matching. Although these MaaS travelers get a lower pay-off, the corresponding operators receive a much higher pay-off. Hence, the matching is still stable because it yields a higher total payoff than the non-MaaS trips.

\begin{table}[H]
\caption{Illustration of stability condition for selected OD pairs}
\label{tab:illustration_stability_condition}
\centering
\setlength\tabcolsep{10pt}
\begin{threeparttable}
\begin{tabular}{cccccccc}
\toprule
\multirow{2}{*}{\textbf{OD}} & \multirow{2}{*}{\textbf{$U_w$}} & \multicolumn{3}{l}{\textbf{MaaS}}               & \multicolumn{3}{l}{\textbf{Non-MaaS}}           \\
                             &                                & Cost & Revenue & Payoff & Cost & Revenue & Payoff \\
\midrule                            
(17, 21) & 37.62 & 32.78 & 17.64 & 22.48 & 30.96 & 15.00 & 21.66 \\
(19, 21) & 27.33 & 24.84 & 14.70 & 17.19 & 23.07 & 12.00 & 16.26 \\
(21, 17) & 37.61 & 32.70 & 17.64 & 22.55 & 30.98 & 15.00 & 21.63 \\
(20, 14) & 36.81 & 36.73 & 22.04 & 22.12 & 35.01 & 18.00 & 19.80 \\
(19, 24) & 35.66 & 34.67 & 19.10 & 20.09 & 33.01 & 16.50 & 19.15 \\
(6, 22)  & 47.01 & 47.01 & 29.39 & 29.39 & 45.43 & 22.58 & 24.16\\
\bottomrule
\end{tabular}
\end{threeparttable}
\end{table}

To further demonstrate the above analysis, in Table~\ref{tab:illustration_stability_condition}, we report the trip utility, traveler costs, operator revenues and total match payoffs for selected OD pairs with the largest gap in the traveler's cost between MaaS and non-MaaS.
We first observe the traveler costs between some OD pairs are pushed to the boundary of non-negative traveler payoff (e.g., (6, 22) ). In other words, the platform squeezes all benefits from the demand side to compensate the supply side and to expand its own profit. It thus leads to a higher traveler cost compared to non-MaaS trips. Nevertheless, these MaaS trips yield much higher revenue for operators. For trips between OD pair (6, 22), the rise in operator revenue is 5.23 compared to the loss in traveler cost of 1.58. Therefore, they have no incentive to deviate from the current match and serve these travelers with non-MaaS trips. Accordingly, the matching is still stable.

So far, $\eta$ is set to optimize the MaaS platform. In Table~\ref{tab:with_vs_without_platform}, we report the system performances with $\eta$ set to optimize other stakeholders, along with the base scenario. Specifically, ``traveler-optimal'' aims to minimize the total traveler cost and ``operator-optimal'' aims to maximize the total operator revenue, while in both cases, $\eta$ is selected to ensure a non-negative MaaS platform profit. 

Compare to the base scenario without MaaS, the average travel cost is reduced in all tested cases. As expected, the maximum saving is achieved in the traveler-optimal scenario with $\eta=0.71$, while the minimum one is obtained in the operator-optimal scenario with $\eta=1.15$. 
In line with the finding in Figure~\ref{fig:MaaS_profit_to_MT_factor}, the MaaS platform gains a much higher profit in the MaaS-optimal scenario, while it breaks even in the other two cases.  
On the other hand, the MT revenue increases significantly in the MaaS-optimal and Operator-optimal cases while the MoD revenue increases in the Traveler-optimal case. 
These findings indicate that MaaS could create a win-win-win situation in the transportation system such that all major stakeholders can benefit from it.

To better understand why MoD revenue remains  the same in both the MaaS-optimal and Operator-optimal scenarios and why the optimal capacity purchase price retains to be equal when $\eta>0.84$ (see Figure~\ref{fig:MaaS_profit_to_MT_factor}), we plot the constraints and feasible region of $p^s$ under the optimal trip fare $p^d$. 
Given that the pricing problem is a linear program, it is easy to see the MaaS platform profit is maximized when $p^s$ stands at the lower bound of its feasible region. 
In both the MaaS-optimal and Operator-optimal scenarios, the MoD revenue constraint is binding, which does not change with the MT capacity price factor. This remains the case as $\eta$ keeps increasing. When $\eta<0.84$, the MT revenue constraint becomes active, thus the optimal $p^s$ decreases with $\eta$. 

\begin{table}[H]
\caption{Comparison of MaaS stakeholder benefits under different settings}
\label{tab:with_vs_without_platform}
\centering
\setlength\tabcolsep{3pt}
\begin{threeparttable}
\begin{tabular}{lcccc}
    \toprule
    & {\multirowcell{2}{\textbf{Without} \\ \textbf{platform}}} & \multicolumn{3}{l}{\textbf{\makecell[l]{With platform}}} \\
    & & \makecell[l]{Traveler-optimal} & \makecell[l]{MaaS-optimal} & \makecell[l]{Operator-optimal} \\
    \midrule
    \textbf{\makecell[l]{Value of $\eta$}}  & - & 0.71  & 0.87 & 1.15 \\
    \midrule
    \textbf{\makecell[l]{Traveler}} & & & & \\
    Avg. travel cost$^*$ &25.53 &23.98 &24.28 &24.40  \\
    & &(-6.07\%) & (-4.90\%)& (-4.40\%)\\
    \midrule
    \textbf{\makecell[l]{MaaS platform}} & & & &\\
    Profit &$-$ & 0 & 104,616.80 & 0\\
    \midrule
    \textbf{\makecell[l]{Operator}} & & & &\\
    MT revenue &638,873.65 &638,873.65 &700,791.19 &850,412.05 \\
     & &(+0.00\%)& (+9.69\%)&(+33.11\%)\\
    MoD revenue &$1.17\times10^6$ & $1.23\times10^6$  & $1.17\times10^6$ & $1.17\times10^6$\\
    & &(+4.97\%)&(0.00\%) &(0.00\%)\\
    \bottomrule   
\end{tabular}
\begin{tablenotes}
    \item[*]\textit{Average travel cost weighted by OD demand.}
\end{tablenotes}
\end{threeparttable}
\end{table}

\begin{figure}[H]
\captionsetup[subfigure]{justification=centering}
\centering
{\includegraphics[width=0.65\linewidth]{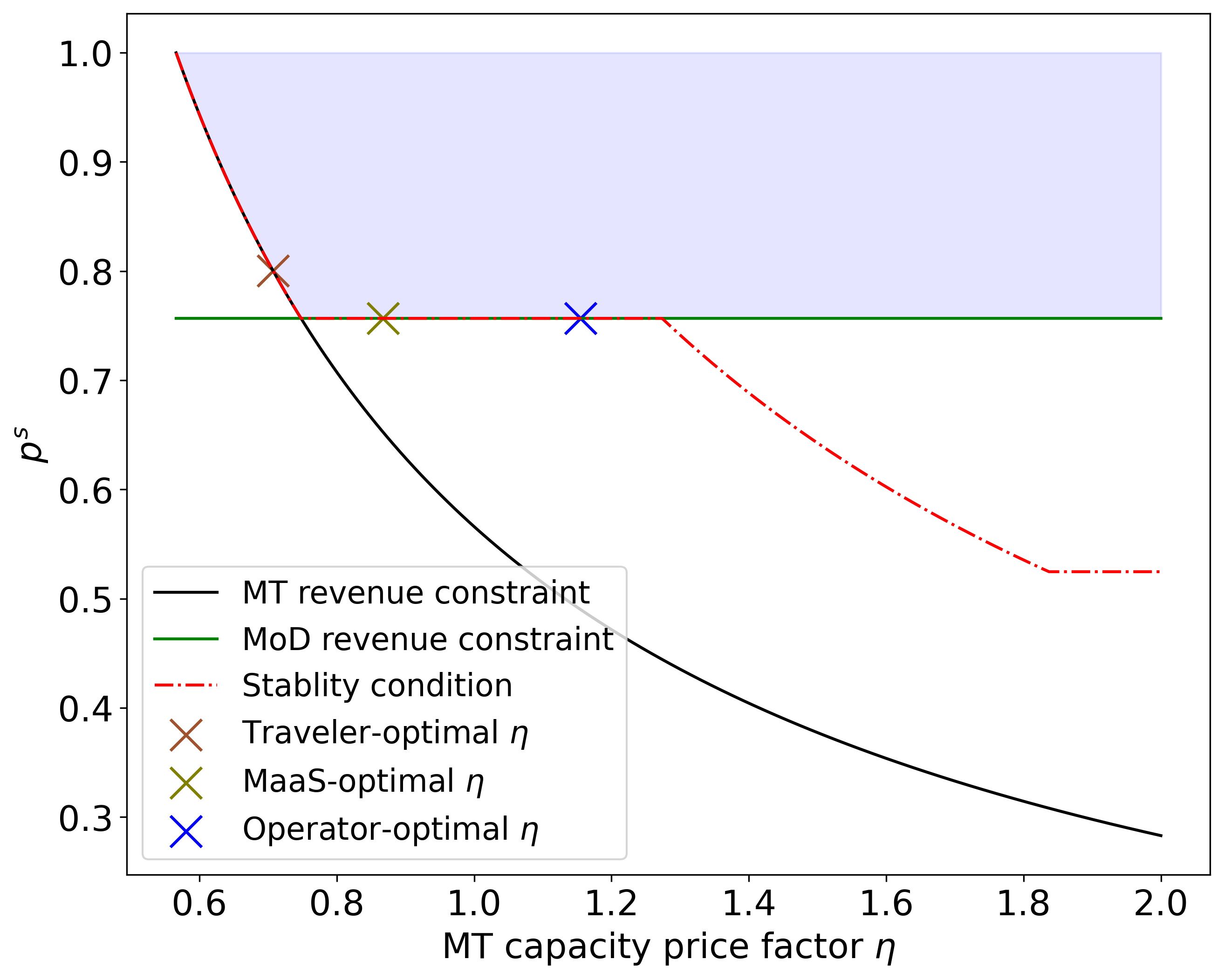}}
\caption{Active bounds for $p^s$ (shaded area corresponds to the feasible region)}\label{fig.active_bound_p_s}
\end{figure}

In addition, under the optimal assignment, the increases in operator revenues are captured by the gap between the feasible region lower bound and the operator revenue constraints. In the traveler-optimal scenario, the optimal price $p^s$ is higher than the MoD revenue constraint, which results in a 4.97\% increases in MoD revenue. In contrast, the MT revenue constraint is binding in this scenario, as also indicated in the same revenue in Table~\ref{tab:with_vs_without_platform}. Similarly, MT revenue is higher (+33.11\% vs +9.69\%) in the Operator-optimal scenario compared to the MaaS-optimal scenario, which is illustrated in Figure~\ref{fig.active_bound_p_s} with larger gap between the the feasible region lower bound and MT revenue constraint.

The more restrictive MoD revenue constraint is due to the optimal assignment outcomes. 
As shown in Table~\ref{tab:assignment_results}, the reduction in non-MaaS MoD trip flows (-65.11\%) is more significant than non-MaaS MT trip flows (-27.59\%). Therefore, the MoD operator requires a large payoff from the MaaS platform to guarantee its revenue, which leads to a higher bound on $p^s$.

\section{Discussion}
\subsection{Non-uniqueness of MaaS many-to-many matching} 
MaaS and non-MaaS travelers share the same physical links in the multi-modal transportation network (e.g., roads and transit lines). Besides, some MoD vehicles (i.e., e-hailing vehicles) and private cars also share the same road network. Hence, Eq.~\eqref{eq:lower_simp} essentially describes a multi-class and multi-modal equilibrium with link interaction, whose uniqueness is in general not guaranteed~\citep{bar2012user, florian2014uniqueness}. 
Nevertheless, we have proved in Proposition~\ref{prop:existence_assignment} that there always exists an optimal MaaS assignment $(q^*,x^*)$, from which we can always obtain a stable outcome $(q^*,x^*,p^*)$ as per Proposition~\ref{prop:existence_stable_outcome}. Further, Corollary~\ref{coro:unique_optimal_pricing} ensures that there is a unique optimal pricing scheme when the optimal assignment satisfies some mild conditions.

When solving the assignment problem, we set the initial solution as one equilibrium in the base scenario without MaaS to mimic the evolution after MaaS is launched in an existing transportation system. This yields one optimal assignment with its corresponding unique pricing scheme. Alternatively, multi-class and (or) multi-modal proportionality assumptions can be imposed to obtain unique link flows for each class and (or) mode~\citep{bar2012user,florian2014uniqueness}. We leave such investigation for future research.

\subsection{Connection to supply chain equilibrium} 
As mentioned in Section~\ref{sec:review}, this study is closely related to the supply chain equilibrium models~\citep[e.g.,][]{nagurney2002supply,nagurney2006relationship}, in which a profit-driven intermediary buys products from manufactures at a wholesale price and sells them to customers at a retail price. Specifically, the MaaS platform modeled in this study plays the role of intermediary in the multi-modal transportation system. However, our model differs from the supply-chain models in two main aspects.

First, instead of elastic supply and demand in supply-chain models, transport capacities and travel demands are considered exogenous in this paper. Correspondingly, operators and travelers are assumed to decide the proportion of capacities and demands using MaaS. 
Second, we adopt stable matching as the concept of stability, instead of classic user equilibrium considered in supply-chain models. In other words, operators and travelers have no incentive to deviate from their current match because they cannot form another match that yields a higher total payoff, rather than a higher individual utility as stated in user equilibrium. Therefore, the stability condition can be interpreted as a relaxation of user equilibrium. Future research could perform an in-depth analysis of the connection between the stability condition of the MaaS model and the equilibrium condition in the supply-chain model.

\section{Conclusion}
How would a MaaS platform survive as an intermediary in a multi-modal transportation system? Or specifically, how does the platform match demand with supply and price on both sides such that both travelers and service operators are willing to participate, meanwhile its own profit is also maximized? 
To answer this question, we adapt the analytical framework of many-to-many stable matching and decompose the complex interactions among travelers, service operators and the MaaS platform into two subproblems.
The first subproblem tackles the optimal assignment of MaaS travel demand and flows. The assignment model captures the route choice of both MaaS and non-MaaS travelers in a congestible multi-modal network subject to the service capacities. We prove that, under mild conditions, the MaaS platform clears the market at equilibrium, i.e., its purchased capacity equals its served travel flows. This result enables us to reformulate the assignment problem as a bilevel program and solve it with a sensitivity-based algorithm. 
The second subproblem regards designing the optimal pricing scheme that ensures the participation of travelers and service operators in MaaS and maximizes the platform's profit. The former is dictated by the stability condition derived from the optimal assignment outcomes. With a simplified but also more practical pricing scheme, we explicitly characterize the stability conditions and show the optimal pricing problem can be reduced to a linear program.We further prove that the resulting problem is always feasible and guarantees a unique optimal solution given each assignment satisfying some mild conditions.

The proposed model is tested on a hypothetical multi-modal transport system based on the Sioux Falls network and compared with a based scenario without MaaS.  The main findings from numerical experiments are summarized as follows:
\begin{itemize}
    \item MaaS has the potential to substitute private driving with public transport even if MaaS travelers are given full freedom in mode and route choice. Specifically, MaaS promotes more multi-modal trips and MT ridership. Accordingly, the capacity utilization of MT increases significantly meanwhile the congestion on the road network is mitigated. These results demonstrate the key benefits of introducing MaaS to an existing transportation system. 

    \item The MaaS platform is able to achieve a positive profit under well-designed pricing schemes. Not only MaaS travelers are better off compared to the base scenario, but also a large fraction of non-MaaS travelers can benefit as MaaS improves the overall system efficiency. Therefore, the average travel cost is reduced. On the other hand, the service operators also enjoy a higher revenue compared to the scenario without MaaS. Consequently, the launch of a MaaS platform creates a win-win-win situation.  
\end{itemize}

There are several directions along which the current model can be extended. First, in this study, we assume the service capacities are exogenous and the non-MaaS fares are fixed. Future research could incorporate the operators' decisions into the model, which could provide additional insights into the cooperation and competition between the MaaS platform and operators. Similarly, the assumption of fixed total demand can be relaxed to analyze the robustness of a MaaS system toward the future growth of travel demand. 
Besides, the numerical results indicate that the performance of MaaS depends on how much the road network and MT network overlap with each other. Hence, it would be interesting to explore the impact of topological properties of the multi-modal transportation network on the social cost saving and profitability of MaaS. 
Finally, as mentioned in the discussion, there exists a connection between matching stability and supply-chain equilibrium. Future research can look into their difference and investigate which one fits better with real practice.

\newpage 
\bibliographystyle{apalike} 
\bibliography{mybib}

\newpage
\appendix

\section{Proof of Lemma~\ref{lemma:market-clearance}}\label{proof:market_clearance}
\begin{proof}
    Note that the capacity purchase price is positive for links with MaaS flows. Let $k_a$ denote the purchased service capacity with feasible price $p_a^{s*}$, and $x_a$ denote the realized MaaS travel flows on link $a$ with feasible link fare $p_{a}^{d*}$, which is distributed from the OD-based MaaS fares.

    The MaaS platform profit maximization problem takes the following abstract form:
    \begin{subequations}
        \begin{align}
            \max_{k_a}\quad & \sum_{a\in \A_r} \left(p_{a}^{d*}x_a - p_a^{s*}k_a \right),\\
            \text{s.t.}\quad & x_a \leq k_a, \; \forall a \in \mathcal{A}_r
        \end{align}
    The corresponding complementarity conditions are:
    \begin{align}
        0 \leq p_a^{s*} - \mu_a \perp k_a \geq 0 \label{eq:comp_market_clearance_k}\\
        0 \leq k_a - x_a \perp \mu_a \geq 0, \label{eq:comp_market_clearance_mu}
    \end{align}
    where $x_a > 0$ implies $k_a > 0$, such that through complementarity condition~\eqref{eq:comp_market_clearance_k}:
    \begin{align}
        \mu_a = p_a^{s*} > 0 
    \end{align}
    consequently, we have by complementarity condition~\eqref{eq:comp_market_clearance_mu}:
    \begin{align}
        k_a = x_a,
    \end{align}
    we conclude that, if MaaS platform aims to maximize its profit, the market is cleared at optimal.
    \end{subequations}
\end{proof}

\section{Proof of Lemma~\ref{lemma:MaaS_UE_cost}}\label{proof:MaaS_UE_cost}
\begin{proof}
    Under the assumptions of Proposition~\ref{prop:existence_assignment}, the complemenarity conditions for the lower-level problem~\eqref{eq:lower_simp_whole} are derived for MaaS travelers as follows:
    \begin{subequations}
        \begin{align}
            0 \leq \left[\pi_{js} + t_{ij} + {\mu}_{ij} - \pi_{is} \right] &\perp x_{ij}^w \geq 0, \; \forall (i, j)\in\A, w \in \W\\
            0 \leq [A_i x^{w} - b_i^w ]&\perp \pi_{is} \geq 0, \forall i \in \N, w \in \W \\
            0 \leq \left[K_a - \sum_{w\in\W} (x_a^w + \tilde{x}_a^w) \right] &\perp {\mu}_a \geq 0,\; \forall a \in \A
        \end{align}     
    where, $A_i$ denote the $i$-th row of the node-link adjacency matrix $A$.

    At optimal, for links with positive flows $x_{ij}^w > 0$, we have:
    \begin{align}
        \pi_{is} = \pi_{js} + t_{ij} + {\mu}_{ij} \label{eq:prop_recursive_ue_cost}.
    \end{align}
    Therefore, for path in the optimal path set $r \in {\R}_w^*$, the UE cost $\pi_w$ can be obtained by setting $\pi_{ss}=0$ and recursively summing over all incident links with Eq.~\eqref{eq:prop_recursive_ue_cost}:
    \begin{align}
       \pi_{w} = \sum_{a \in r} (t_a^* + {\mu}_a^*)
    \end{align}
    \end{subequations}
\end{proof}

\section{Detailed experiment setup}\label{sec:SF_link_parameter}

The road network free-flow travel times are set to be the same as Sioux Falls benchmark~\citep{bargeraNetwork}, while we reduce the link capacity by 25\% to avoid modeling the exogenous vehicular traffic. 
The MT link travel time $T_a$ is set to be 1.6 times the road link free-flow travel time. 
For simplicity, we assume link-based non-MaaS MoD fare $\tilde{p}_a = T_a$, private driving cost $\tilde{p}_a = 1.8T_a$, and MT fare $\tilde{p}_a = 0.5 T_a$
The values of regular link parameters are reported in Table~\ref{tab:SF_link_parameter}.

For MoD, the matching coefficient in Eq.~\eqref{eq:ta_MoD_access} is set to $\kappa_m = 26.78 \text{km}^2/\text{min}$, 
and the minimum vacant vehicle time is set to $\underline{V}_m = 0.5$ minutes. Since MoD provides door-to-door trips, we simply set the egress time as zero.
For MT, the access and egress times are set as 1.25 and 0.25 minutes, respectively. 
As mentioned in Section~\ref{subsec:results_setup}, we simplify the network by merging transfer links and access links. Hence, an extra transfer time of $T_0 = 1.0$ minute is added to all access links. Besides, a planning cost $P_0$ equivalent to $2.5$ minutes is added to each access link in a non-MaaS multi-modal trip to capture the extra cost in self-planned mode transfers. 

Regarding the exogenous service capacity, we set the MoD total vehicle time as $2 \times 10^6$ minutes and assume  each MT line can serve up to $15,000$ passengers. These values are chosen such that the MoD and MT service capacities are not fully saturated in the base scenario without MaaS.

The parameters in Algorithm~\ref{alg:assignment} are well-tuned and set as follows: 

\begin{itemize}
    \item Step sizes: $\alpha = 2.5\times 10^{-4}, \gamma = 2.5\times 10^{-4}, \beta=0.2, N=200$;
    \item Augmented Lagrangian parameters: $\rho = 1$ (with additional upper bound of $\rho^{(k)} \leq 200$), $\phi=0.1, \sigma=0.85$;
    \item Stopping thresholds $\epsilon_{q}=\epsilon_{z}=1\times 10^{-7}$.
\end{itemize}
The algorithm is implemented in Python and the backward propagation is performed by PyTorch. Finally, the pricing problem is solved by Gurobi.

\begin{longtable}{cccccccc}
\caption{Parameter values of regular links (bidirectional)}\label{tab:SF_link_parameter}\\ \toprule
\multirow{2}{*}{\textbf{Link}} & \multicolumn{4}{l}{\textbf{Road network}}                                                                                                                                 & \multicolumn{3}{l}{\textbf{MT network}}                                                        \\
                      & $T_a$ & $K_a $    &  $\rho_a^{\text{Car}}$ & $\rho_a^{\text{MoD}}$ & $T_a$ & $K_a$  & $\rho_a^{\text{MT}}$ \\
                      \midrule
\endfirsthead
\caption{Parameter values of regular links (bidirectional) (Cont.)}\\
\toprule
\multirow{2}{*}{\textbf{Link}} & \multicolumn{4}{l}{\textbf{Road network}}                                                                                                                                 & \multicolumn{3}{l}{\textbf{MT network}}                                                        \\
                      & $T_a$ & $K_a $    &  $\rho_a^{\text{Car}}$ & $\rho_a^{\text{MoD}}$ & $T_a$ & $K_a$  & $\rho_a^{\text{MT}}$ \\
                      \midrule
\endhead
\hline
\multicolumn{8}{l}{continued on next page}\\   \bottomrule
\endfoot
\bottomrule
\endlastfoot
(1, 2)   & 6.00  & 19425.15 & 10.80 & 6.00  & -    & -        & -    \\
(1, 3)   & 4.00  & 17552.60 & 7.20  & 4.00  & -    & -        & -    \\
(2, 6)   & 5.00  & 3718.64  & 9.00  & 5.00  & 8.00 & 15000.00 & 2.50 \\
(3, 4)   & 4.00  & 12832.89 & 7.20  & 4.00  & 6.40 & 15000.00 & 2.00 \\
(3, 12)  & 4.00  & 17552.60 & 7.20  & 4.00  & -    & -        & -    \\
(4, 5)   & 2.00  & 13337.10 & 3.60  & 2.00  & 3.20 & 15000.00 & 1.00 \\
(4, 11)  & 6.00  & 3681.62  & 10.80 & 6.00  & 9.60 & 15000.00 & 3.00 \\
(5, 6)   & 4.00  & 3711.00  & 7.20  & 4.00  & 6.40 & 15000.00 & 2.00 \\
(5, 9)   & 5.00  & 7500.00  & 9.00  & 5.00  & 8.00 & 15000.00 & 2.50 \\
(6, 8)   & 2.00  & 3673.94  & 3.60  & 2.00  & 3.20 & 15000.00 & 1.00 \\
(7, 8)   & 3.00  & 5881.36  & 5.40  & 3.00  & -    & -        & -    \\
(7, 18)  & 2.00  & 17552.60 & 3.60  & 2.00  & -    & -        & -    \\
(8, 9)   & 10.00 & 3787.64  & 18.00 & 10.00 & -    & -        & -    \\
(8, 16)  & 5.00  & 3784.37  & 9.00  & 5.00  & 8.00 & 15000.00 & 2.50 \\
(9, 10)  & 3.00  & 10436.84 & 5.40  & 3.00  & 4.80 & 15000.00 & 1.50 \\
(10, 11) & 5.00  & 7500.00  & 9.00  & 5.00  & 8.00 & 15000.00 & 2.50 \\
(10, 15) & 6.00  & 10134.00 & 10.80 & 6.00  & 9.60 & 15000.00 & 3.00 \\
(10, 16) & 4.00  & 3641.19  & 7.20  & 4.00  & 6.40 & 15000.00 & 2.00 \\
(10, 17) & 8.00  & 3745.13  & 14.40 & 8.00  & -    & -        & -    \\
(11, 12) & 6.00  & 3681.62  & 10.80 & 6.00  & 9.60 & 15000.00 & 3.00 \\
(11, 14) & 4.00  & 3657.38  & 7.20  & 4.00  & 6.40 & 15000.00 & 2.00 \\
(12, 13) & 3.00  & 19425.15 & 5.40  & 3.00  & -    & -        & -    \\
(13, 24) & 4.00  & 3818.44  & 7.20  & 4.00  & 6.40 & 15000.00 & 2.00 \\
(14, 15) & 5.00  & 3845.64  & 9.00  & 5.00  & -    & -        & -    \\
(14, 23) & 4.00  & 3693.59  & 7.20  & 4.00  & 6.40 & 15000.00 & 2.00 \\
(15, 19) & 3.00  & 10923.56 & 5.40  & 3.00  & -    & -        & -    \\
(15, 22) & 3.00  & 7199.39  & 5.40  & 3.00  & 4.80 & 15000.00 & 1.50 \\
(16, 17) & 2.00  & 3922.43  & 3.60  & 2.00  & 3.20 & 15000.00 & 1.00 \\
(16, 18) & 3.00  & 14759.92 & 5.40  & 3.00  & 4.80 & 15000.00 & 1.50 \\
(17, 19) & 2.00  & 3617.96  & 3.60  & 2.00  & 3.20 & 15000.00 & 1.00 \\
(18, 20) & 4.00  & 17552.60 & 7.20  & 4.00  & -    & -        & -    \\
(19, 20) & 4.00  & 3751.96  & 7.20  & 4.00  & 6.40 & 15000.00 & 2.00 \\
(20, 21) & 6.00  & 3794.93  & 10.80 & 6.00  & 9.60 & 15000.00 & 3.00 \\
(20, 22) & 5.00  & 3806.77  & 9.00  & 5.00  & -    & -        & -    \\
(21, 22) & 2.00  & 3922.43  & 3.60  & 2.00  & 3.20 & 15000.00 & 1.00 \\
(21, 24) & 3.00  & 3664.02  & 5.40  & 3.00  & 4.80 & 15000.00 & 1.50 \\
(22, 23) & 4.00  & 3750.00  & 7.20  & 4.00  & -    & -        & -    \\
(23, 24) & 2.00  & 3808.88  & 3.60  & 2.00  & 3.20 & 15000.00 & 1.00 \\
\end{longtable}

\section{Performance of Algorithm~\ref{alg:assignment}}\label{sec:algo_performance}

The algorithm converges with both gaps below the stopping thresholds: $g_q = 9.56 \times 10^{-8}$ and $g_z = 9.97 \times 10^{-8}$. 
Figures~\ref{fig:outer_convergence} and \ref{fig:inner_convergence} plot the two gaps against the number of iterations. For a better illustration, the results are also smoothed by exponential averaging.  

As can be seen in the plots, both gaps drop significantly within the first several hundreds of iterations, then continue to reduce approximately linearly until convergence. A lot of oscillations, however, are observed due to the coordinated descent feature of the algorithm~\citep{tseng2001convergence}. 

\begin{figure}[H]
\captionsetup[subfigure]{justification=centering}
\centering
{\includegraphics[width=0.8\linewidth]{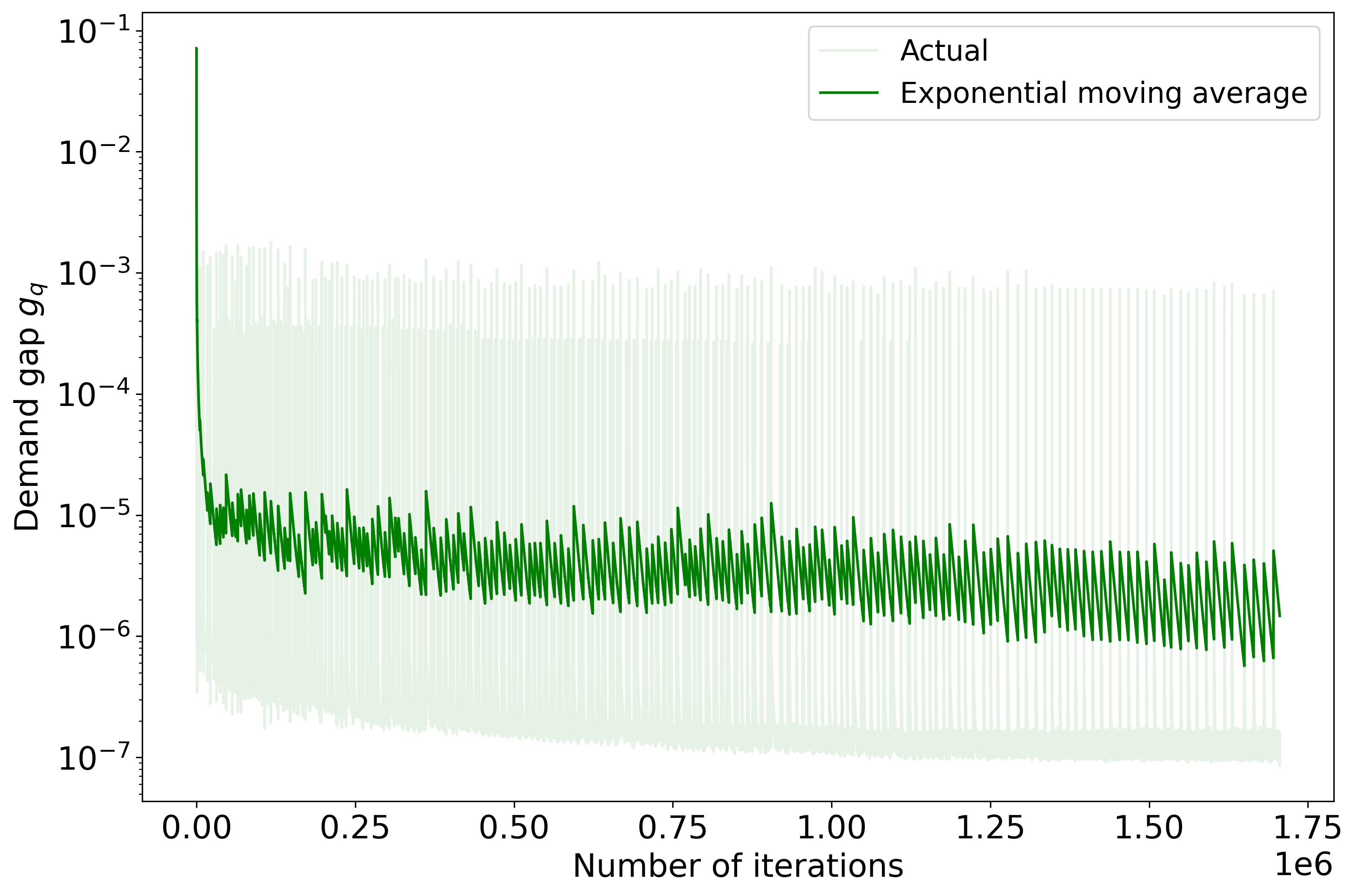}}
\caption{Convergence of demand gap.}
\label{fig:outer_convergence}
\end{figure}

\begin{figure}[H]
\captionsetup[subfigure]{justification=centering}
\centering
{\includegraphics[width=0.8\linewidth]{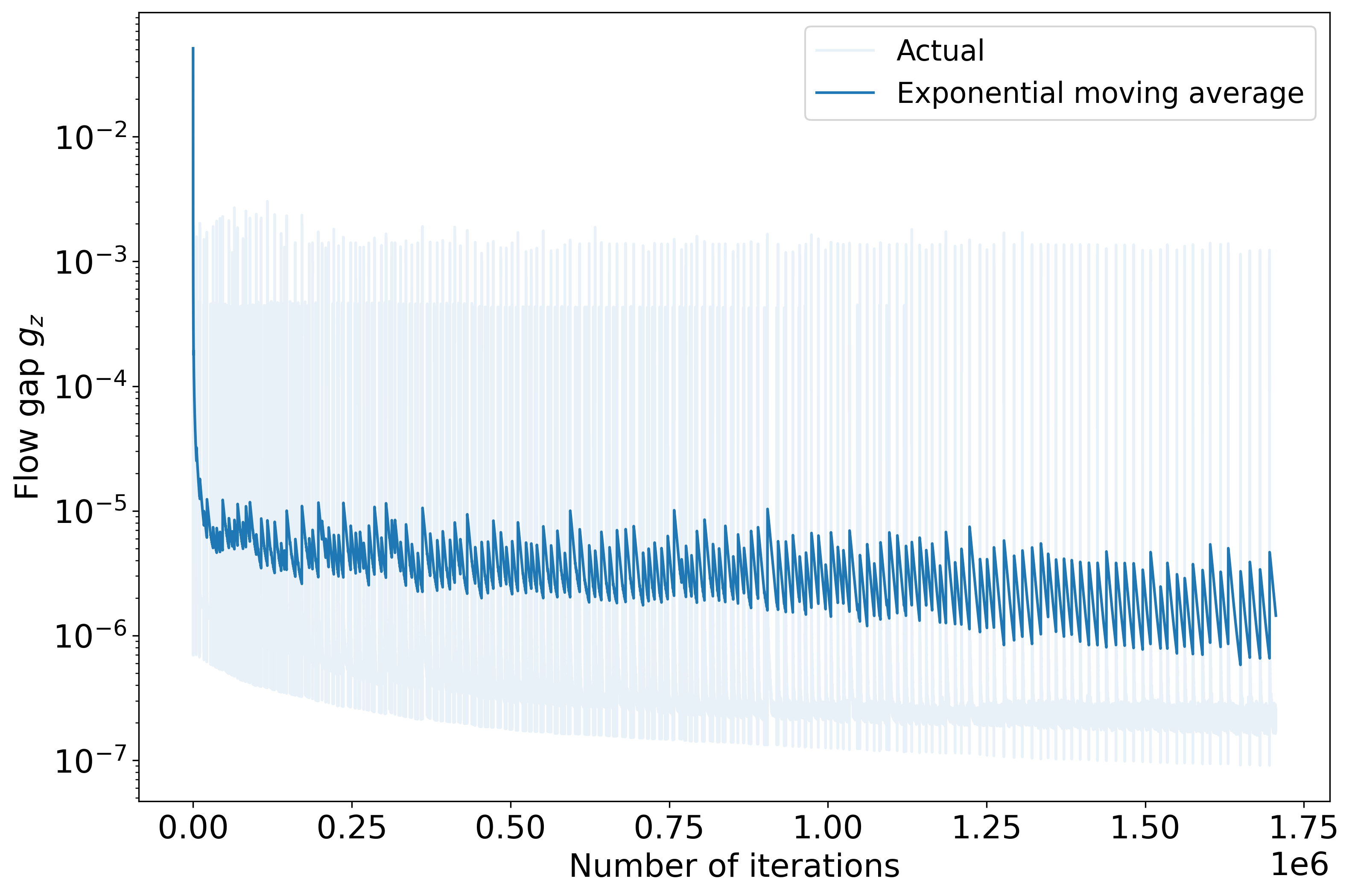}}
\caption{Convergence of flow gap.}
\label{fig:inner_convergence}
\end{figure}

\section{Analysis of compensated OD pairs}\label{sec:compensated_OD_pair}
The compensated MaaS trips under the optimal pricing scheme are between the following pairs of nodes (bidirectional):
\begin{itemize}
    \item (1, 3), (3, 12), (7, 8), (12, 13)
\end{itemize}
As can be seen in Figure~\ref{fig:SF_MaaS_net}, all these OD pairs are connected by a single road link and have no convenient access through MT. Under the optimal assignment, all MaaS travelers take MoD trips while non-MaaS travelers drive alone, both via the shortest paths. Therefore, there is no extra traffic induced on the road network when switching these travelers between MaaS and non-MaaS. However, the MaaS platform chooses to capture and even compensate these travelers because their shifts in mode benefit the entire system. This is done by squeezing the MoD vacant vehicle time (see Eq.~\eqref{eq:MoD_cap}), which lowers the MoD service quality and in turn, encourages the MT ridership.

\end{document}